\documentclass[12pt]{article}

\usepackage[english]{babel}

\usepackage[a4paper,top=2cm,bottom=2cm,left=3.1cm,right=3.1cm,marginparwidth=1.75cm]{geometry}

\usepackage{amsmath}
\usepackage{amssymb} 
\usepackage{graphicx}
\usepackage[hidelinks]{hyperref}
\usepackage{enumerate}
\usepackage{xcolor}
\usepackage{soulpos}

\ulposdef{\hlst}{%
    \rlap{\textcolor{yellow}{\rule[-.75ex]{\ulwidth}{2.5ex}}}%
    \rule[.45ex]{\ulwidth}{.1ex}%
}

\usepackage{amsthm}
\theoremstyle{definition}
\newtheorem{definition}{Definition}[section]
\newtheorem{theorem}[definition]{Theorem}
\newtheorem{notation}[definition]{Notation}

\usepackage{tikz-cd}
\usepackage{subcaption}
\usepackage[labelfont=bf]{caption}
\usepackage{tikz}
\usepackage{pgfplots}
\usetikzlibrary{decorations.markings, graphs, fit, calc}
\pgfplotsset{compat=1.17}

\usepackage{booktabs}
\usepackage{multirow}
\usepackage{longtable}

\tikzset{
half fill/.style 2 args={fill=#2, path picture={
\fill[#1, sharp corners] 
    (path picture bounding box.south) rectangle
    (path picture bounding box.north west) -- cycle;}}}

\makeatletter
\newcommand{\circlearrow}{}% just in case
\DeclareRobustCommand{\circlearrow}{%
  \mathrel{\vphantom{\rightarrow}\mathpalette\circle@arrow\relax}%
}
\newcommand{\circle@arrow}[2]{%
  \m@th
  \ooalign{%
    \hidewidth$#1\circ\mkern1mu$\hidewidth\cr
    $#1\longrightarrow$\cr}%
}
\makeatother

\usepackage{apacite}
\usepackage{doi}

\title{Studying self-organisation across the biosphere with process-enablement graphs}
\author{\small Emmy Brown$^a$ and Sean T. Vittadello$^{a,b,c,*}$ \\[5pt]
\small $^a$School of Mathematics and Statistics, The University of Melbourne, \\ \small Parkville VIC 3010, Australia \\[5pt]
\small $^b$School of BioSciences, The University of Melbourne, \\ \small Parkville VIC 3010, Australia \\[5pt]
\small $^c$ARC Centre of Excellence for the Mathematical Analysis of Cellular Systems \\[5pt]
\small $^*$Corresponding author: \textit{sean.vittadello@unimelb.edu.au}}
\date{\normalsize \today}

\begin{document}
\maketitle \vspace{-2.5em}
\begin{abstract}
At the heart of many contemporary theories of life is the concept of biological self-organisation: organisms have to continuously produce and maintain the conditions of their own existence in order to stay alive. The way in which these accounts articulate this concept, however, differs quite significantly. As a result, it can be difficult to identify self-organising features within biological systems, and to compare different descriptions of such features. In this paper, we develop a graph theoretic formalism -- process-enablement graphs -- to study the organisational structure of living systems. Cycles within these graphs capture self-organising components of a system in a general and abstract way. We build the mathematical tools needed to compare biological models as process-enablement graphs, facilitating a comparison of their corresponding descriptions of self-organisation in a consistent and precise manner. We apply our formalism to a range of classical theories of life and demonstrate exactly how these models are similar, and where they differ, with respect to their organisational structure. While our current framework does not demarcate living systems from non-living ones, it does allow us to better study systems that lie in the grey area between life and non-life.

\phantom{}

\end{abstract}

\section{Introduction}

Several attempts have been made throughout the past century to construct theories of life, including $(M,R)$-systems \cite{rosen1991life}, autopoiesis \cite{maturana1980autopoiesis}, the chemoton \cite{ganti2003principles}, autocatalytic sets \cite{kauffman1986autocatalytic}, and constraint closure \cite{montevil2015biological}. These accounts are each slightly different from one another, and each emphasises different characteristics of living systems. However, in one way or another they all explore the concept of \textit{organisational closure}, which is the idea that biological systems have to continuously produce and maintain the conditions of their own existence in order to stay alive. In other words, they have to self-organise and self-determine. More precisely, a system is said to realise \textit{organisational closure} if each of its processes is generated by some other process in the system, but also if each process contributes to the generation of some other process in the system \shortcite{jaegernaturalizing, montevil2015biological, moreno2015biological}. As a result, any process in a system that realises organisational closure could not operate without the rest of the processes in the system also unfolding; that is, the process could not exist in isolation.

Recent attempts to formalise and investigate organisational closure have focussed on an understanding of organisational closure as the \textit{closure of constraints} \cite{montevil2015biological, juarrero2023context, nave2025drive}. Here, a constraint is a boundary condition that impinges on relevant processes, and reduces their degrees of freedom. For example, our vasculature acts as a constraint on the flow of blood through our body since, without blood vessels, blood would flow in all directions instead of being neatly transported to our cells. Over different time scales, however, constraints also change like any other process. Using the present example, over longer timescales, our vasculature will change its physical structure through processes like angiogenesis. As such, constraints are themselves processes and therefore, systems that realise a closure of constraints can be considered as a special case of the more general set of systems that realise a closure of processes.

Many non-biological systems like candles, the water cycle, tornadoes and Bénard cells realise a closure of processes, but they do not necessarily realise a closure of constraints \cite{cusimano2020objectivity, mossio2017makes}. Self-organising systems like these (i.e. those that realise a closure of processes) are the precursors for the kinds of organisational dynamics found in self-determining systems, which realise a closure of constraints \cite[p.101]{juarrero2023context}. As shown in \citeA{el2020life}, a better understanding of how systems realise a closure of processes can help us to understand the transition from physicochemical systems to life-constrained, ecological ones. Therefore, a framework to study and analyse how a closure of processes is realised in systems could prove to be valuable for studying both the origins of life and biological organisation in general. We wish to outline such a framework in the remainder of this paper.

To do so, we employ the approach for representing models as simplicial complexes introduced in the model comparison framework developed by \citeA{vittadello2021model, vittadello2022group}, but translate it into the language of directed graphs so that we can model asymmetrical relationships between processes, as opposed to symmetrical interconnections between concepts.

In particular, we consider how processes \textit{enable} one another within a particular model and give a precise definition of enablement. More generally, we define a particular type of graph homomorphism that preserves cycles of processes between models, allowing for fine-grained perspectives to be consistently compared to coarse-grained ones. This contextualises the problem of self-organising features seeming to appear or disappear depending on the level of analysis \cite{cusimano2020objectivity, nahas2023s}. We also avoid the problem of constraints at one time scale becoming processes at another. This allows us to consistently compare the realisation of organisational closure from many different perspectives. In particular, our framework allows lower-level models of biochemical interactions to be compared with more abstract models of the cell, and we can explicitly see when and how organisational closure is achieved in these systems. We give an illustrative example of this aspect of our framework in Section \ref{sect: FA}.

The remainder of this paper is organised as follows. In Section \ref{sect: p&e}, we develop the concepts of \textit{process} and \textit{(direct) enablement} in relation to the biosphere. In Section \ref{sect: $pe$-graphs}, we then use these concepts to define the central object of study in this paper, namely the \textit{process-enablement graph}. We also put forward a mathematical conception of organisational closure in terms of cycles in process-enablement graphs, and develop mathematical tools to consistently compare such cycles between graphs. Section \ref{sect: applications} explores a range of applications of our formalism and explains how to translate abstract models of living systems into the language of process-enablement graphs. We use each example to highlight features of our mathematical framework. Finally, in the Conclusion we note that the formalism presented in this paper is static, not dynamic, and will need to be extended appropriately in order to adequately model evolution, development, the origins of life and biological agency.

\section{Process and enablement} \label{sect: p&e}

The concept at the core of this paper is that of a \textit{process-enablement graph}, so we begin with a discussion of process, enablement, and related notions.
\begin{definition}[Process]
We consider a \textit{process} as an organised sequence of consecutive physicochemical changes that can involve production, destruction, alteration, or movement.
\end{definition}
Our definition of a process is necessarily very general, and is intended to include the myriad context-dependent notions of process in the literature.
\begin{definition}[System]
We call a set of processes interacting contemporaneously with each other a \textit{system}. A subset of the processes within a system which also interact contemporaneously with each other is called a \textit{subsystem}.
\end{definition}
The requirement that the processes in a system \textit{interact} prevents us from placing two very distant processes together into a single system. For instance, we could not say that a particular tree in the Amazon rainforest together with a particular earthworm in the soil here in Melbourne constitute a single system, since neither one will be affected by the absence of the other. However, both organisms are part of the global network of processes within the system of Earth's biosphere, and they may interact (albeit very indirectly) through this network. 

Although these notions of process and system are quite abstract, processes at all levels of the biosphere will be familiar to biologists: biochemical reactions, protein folding, transcription, translation, cell division, action potentials, predation, natural selection, inhalation of oxygen, and digestion of food. Examples of systems should also be familiar, such as: cells, organs, organisms, and ecosystems. Systems themselves are also processes as they are organised sequences of consecutive physicochemical changes. But by using the term \textit{system} we are emphasising that we are thinking of a \textit{network} of interacting processes.

\begin{notation}
Given any process $p$, we can associate it with the region of space $X$ which is occupied by $p$. The space $X$ is generated by $p$ and is inherently tied to it. That is, as the process $p$ interacts with its environment, the space $X$ associated with $p$ will also change. In any case, this assignment allows us to say that $p$ is happening within the region $X$ and is not happening outside of $X$. When required, we write $p_{_X}$ for a process $p$ with corresponding spatial region $X$. 
\end{notation}

To capture the hierarchical nature of processes in biology we introduce the following notion.
\begin{definition}[Process happening within another process]
Let $p_{_X}$ and $q_{_Y}$ be two processes. We say that $p_{_X}$ \textit{happens within} $q_{_Y}$ if $p_{_X}$ and $q_{_Y}$ are occurring contemporaneously and $X$ is contained within $Y$.
\end{definition}
For example, within an adult human, a single red blood cell passing through a capillary is a process that happens within the broader process of blood circulation. 

We now consider the concept of enablement. The notion of enablement has been developed as an important way to understand how the dynamics of the living biosphere differ from the world modelled by physics \shortcite{longo2013extended, kauffman2019world}. In short, biological processes often provide the conditions needed for something else to occur without necessarily causing the latter. For example, the process of trees growing in a forest enables birds to build nests inside them, but the trees don't `cause' the birds to do so. We define enablement precisely as follows.
\begin{definition}[Enablement]
Within a system $S$, we say that a process $p$ \textit{enables} another distinct process $q$ when $p$ has provided some of the necessary conditions for $q$ to occur. The necessity should be understood counterfactually. In other words, if $p$ had not happened, then $q$ would not be happening.
\end{definition}
Note that enablement is defined with respect to the given system $S$, so that $p$ may not be necessary for $q$ to occur within some other system $T$. For example, within the cell, ribosomes are necessary to produce polypeptides from RNA sequences via translation, so the diffusion of ribosomes throughout the cell enables polypeptide synthesis. But in the lab we can make polypeptides synthetically via chemical ligation, without the need for ribosomes, so in this case polypeptide synthesis is not enabled by ribosome diffusion.
\begin{definition}[Direct enablement]
If $p_{_X}$ enables $q_{_Y}$ and there also exists a time interval $\tau$ over which $p_{_X}$ and $q_{_Y}$ are interacting with each other, then we say that $p_{_X}$ enables $q_{_Y}$ \textit{directly}, and we write this as $p_{_X} \to q_{_Y}$ or $p \to q$. By `interacting' we mean that $X$ and $Y$ intersect over $\tau$, and within this spatiotemporal region there is some physical interaction between the two processes. 
\end{definition}
Using the example from the previous paragraph, ribosome production directly enables ribosome diffusion through the cell which in turn directly enables translation. Conversely, in a particular forest, a tree releasing a seed might enable a new plant to grow once it rains. However, this enablement is not direct since there is an intermediate period where the tree is not releasing the seed and the new plant is also not growing. We focus on direct enablements in this paper because they reflect a tighter organisation between processes when compared to indirect enablements. We therefore hypothesise that direct enablements will be a more fruitful tool for studying the organisation of living systems, though we do not justify this claim further here.

To summarise, in order to investigate whether a direct enablement $p_{_X} \to q_{_Y}$ within a system $S$ is genuine or not, we need to check two things. First, we need to see whether the hypothetical removal of $p_{_X}$ from $S$ would result in $q_{_Y}$ not occurring. This ensures that $p_{_X}$ enables $q_{_Y}$. Second, we need to check that there is a time interval $\tau$ over which $p_{_X}$ and $q_{_Y}$ are interacting with each other. This ensures that the enablement is direct.

\section{Mathematical framework} \label{sect: $pe$-graphs}

In this section we introduce and develop our mathematical framework based on our central notion of a process-enablement graph.

\subsection{Process-enablement graphs}
We begin with the required background on graph theory. A \textit{directed graph} $G = (V(G), E(G))$ is a pair of finite sets where $V(G)$ is a nonempty set of \textit{vertices} and $E(G) \subseteq V(G)^2$ is a set of \textit{(directed) edges}. For notational simplicity, we denote an edge $(u,v) \in E(G)$ by $uv$. A directed graph is visualised by representing the vertices as labelled points and the edges as arrows. For example, the edge $uv$ is drawn as $u \to v$. This definition of a directed graph allows for \textit{loops}, which are edges from a vertex to itself, for example $u \to u$, but not parallel edges, which are distinct edges between the same two vertices pointed in the same direction, for example $u \rightrightarrows v$. For the remainder of this paper we refer to directed graphs simply as graphs.

A \textit{subgraph} of $G$ is a graph $H$ such that $V(H) \subseteq V(G)$ and $E(H) \subseteq E(G)$. We say that $G$ contains $H$ and write $H \subseteq G$. If $V(H) \neq V(G)$ or $E(H) \neq E(G)$ then $H$ is a \textit{proper subgraph} of $G$. If $S$ is a nonempty subset of $V(G)$, then the subgraph $G[S]$ of $G$ \textit{induced} by $S$ has vertex set $S$ and edge set $E(G[S]) = \{ uv \in E(G) \mid u,v \in S \}$. A subgraph $H$ of $G$ is an \textit{induced subgraph} if there exists a nonempty subset $A \subseteq V(G)$ such that $H = G[A]$. 

A \textit{walk} $W$ of length $n$ from $v_0$ to $v_n$ is an alternating sequence of vertices and edges $v_0, e_0, v_1, e_1, \ldots, e_{n-1}, v_n$ where $e_i \in \{v_iv_{i+1}, v_{i+1}v_i\} \cap E(G)$ for all $i=0, \ldots, n-1$. If $e_i = v_iv_{i+1}$ for all $i=0, \ldots, n-1$ then $W$ is \textit{directed}. We call $v_0$ the \textit{initial} vertex of $W$, $v_n$ the \textit{final} vertex of $W$, and we say that $W$ is a walk from $v_0$ to $v_n$. If the vertices in $W$ are distinct, then $W$ is a \textit{path}. A \textit{cycle} $C$ of length $n$ is a directed walk of length $n$ such that all vertices are distinct except for the initial and final vertices. Note that a cycle of length $1$ is a loop. A graph $G$ is \textit{connected} if there exists a path from $u$ to $v$ for all pairs of vertices $u, v \in V(G)$.

To compare the structures of two graphs $G$ and $H$ we can consider a \textit{graph homomorphism} $\phi \colon G \to H$, which is a structure preserving map given by a vertex function $\phi \colon V(G) \to V(H)$ such that if $uv \in E(G)$ then $\phi(u)\phi(v) \in E(H)$. We call $G$ the \textit{source} graph and $H$ the \textit{target} graph of $\phi$. For a subgraph $G_0 \subseteq G$ the \textit{induced image} of $G_0$ under $\phi$ is the induced subgraph $H[\phi(V(G_0))]$ of $H$, which we denote by $H[\phi(G_0)]$. A more general notion of a structure preserving map is a \textit{weak graph homomorphism} $\phi \colon G \to H$, which is a vertex function $\phi \colon V(G) \to V(H)$ such that if $uv \in E(G)$ and $\phi(u) \neq \phi(v)$ then $\phi(u)\phi(v) \in E(H)$. Weak graph homomorphisms allow the contraction of edges into vertices, while otherwise preserving the edge structure \cite{hell2004graphs}. Within our context, weak graph homomorphisms allow us to zoom out from lower-level models of a system to view the system more coarsely.

We now define our notion of a \textit{process-enablement graph} for modelling physical systems from a process perspective.

\begin{definition}[Process-enablement graph]
A \textit{process-enablement graph}, or \textit{$pe$-graph}, is a connected graph $G$ where $V(G)$ is a set of processes occurring contemporaneously in a particular system $S$ over a particular time interval $\tau$, and $E(G)$ is the set of direct enablements between the processes.
\end{definition}
We require $pe$-graphs to be \textit{connected} graphs in order to satisfy the \textit{interacting} condition for a system.

\subsection{Organisational closure}

In this section, we define organisational closure in terms of a closure of processes and demonstrate that the fundamental building block of systems that realise organisational closure are cycles of processes. We begin by giving a formal definition of organisational closure, modified from Definition 4 in \citeA{montevil2015biological}, in order to fit the context of $pe$-graphs.

\begin{definition}[Organisational closure]
    A $pe$-graph $G$ realises \textit{organisational closure} if for each process $p \in V(G)$ there exist processes $q$, $r \in V(G)$ such that $qp$, $pr \in E(G)$, that is we have the chain of direct enablements $q \to p \to r$. If $G$ realises organisational closure then we simply say that $G$ is \textit{closed}. $G$ is \textit{strictly} closed if $G$ is closed and no proper subgraph of $G$ is closed.
\end{definition}
Note that a $pe$-graph that is a cycle is closed. We now aim to understand organisational closure in terms of cycles in $pe$-graphs. To do this, we first show that every closed $pe$-graph must contain a cycle.

\begin{theorem}
    If $G$ is a closed $pe$-graph then it must contain a cycle.
    \label{thm: closedhavecycles}
\end{theorem}
\begin{proof}
    Let $v_0 \to v_1 \to \cdots \to v_n$ be a path of maximum length $n$ in $G$. Since $G$ is closed, there exists $u \in V(G)$ such that $u v_0 \in E(G)$. Then $u = v_i$, for some $i = 0, \ldots, n$, since paths in $G$ have maximum length $n$. So, if $i = 0$ then $u \to v_0$ is a cycle, and if $i = 1, \ldots, n$ then $u \to v_0 \to \cdots \to v_i$ is a cycle.
\end{proof}

Next, we classify strictly closed graphs as cycles, which will allow us to greatly simplify our investigation of organisational closure.

\begin{theorem}[Classification of strictly closed $pe$-graphs]
    A $pe$-graph $G$ is strictly closed if and only if it is a cycle.
    \label{thm: classificationstrictclosed}
\end{theorem}
\begin{proof}
    If $G$ is strictly closed then, since $G$ is closed, it contains a cycle $C$ by Theorem \ref{thm: closedhavecycles}. Since no proper subgraph of $G$ is closed, $C$ must be the entire graph $G$. Conversely, suppose $G$ is a cycle. Then $G$ is closed. Further, if $H \subseteq G$ is a subgraph of $G$ then $H$ is closed if and only if $H = G$, so no proper subgraph of $G$ is closed. It follows that $G$ is strictly closed. 
\end{proof}

The classification described in Theorem~\ref{thm: classificationstrictclosed} strongly indicates that the `correct' object of analysis for investigating organisational closure is the cycles formed by processes, not just the general and abstract relationships between individual processes. The same conclusion also applies for processes \textit{qua} constraints, as the proof given for Theorem \ref{thm: classificationstrictclosed} relies only on graph theoretic concepts, not processes themselves. As such, an understanding of the cycle structure of process-enablement graphs indicates exactly how and where organisational closure is realised in a particular system.

Further evidence for the primacy of cycles is also provided by Theorem \ref{thm: characterisationclosed} in the Appendix, where we classify closed $pe$-graphs as graphs which only contain cycles connected by directed paths. However, for our purposes, we follow \citeA{montevil2015biological} and find the concept of strict closure to be more useful for our applications in Section \ref{sect: applications}. We will therefore use the term `closure' to refer to strict closure for the remainder of this paper.

\subsection{Perspectival modelling}

There are two key advantages to using $pe$-graphs to study biological self-organisation. First, $pe$-graphs allow us to quickly identify self-organising components within a model of a biological system by examining the cycles within $pe$-graphs. This feature is not particularly unique to our graph theoretic formalism, and there are many alternative methods to formalise, identify and study self-organisation \cite{rosen1991life, maturana1980autopoiesis, montevil2015biological}. However, our graph-theoretic framework is particularly well suited to studying structural properties, such as cycles of processes, within physical systems. Second, the significant advantage of expressing biological models as graphs of processes and enablements is that we can construct graph homomorphisms between $pe$-graphs. This allows us to consistently compare the realisation of self-organisation from different perspectives. We will provide examples of this process in Section \ref{sect: applications}.

It is important to note that every $pe$-graph represents a particular perspective of a system. As a result, for any given system we could draw many distinct $pe$-graphs that each represent a different way of looking at the same system. For instance, if two scientists observe a bird in a forest, a biochemist would likely draw a very different $pe$-graph to describe it compared with an ecologist. This is a fundamental feature of the nature of scientific modelling. Every model represents a particular point of view and a particular way of thinking. Every model leaves things out and distorts reality in one way or another. We agree with \citeA{massimi2022perspectival} that it is only through comparing many models together that we can really begin to understand nature. This doesn't mean that every $pe$-graph is a valid representation of a system. Each time you draw an arrow in a $pe$-graph you are claiming the existence of a genuine lawlike dependency that could be investigated. Thus, whether a $pe$-graph is accurate or not is ultimately up to empirical investigation, not just the whim of a particular scientist.

However, if we want to construct a framework that can accurately compare self-organisation across different perspectives, the graph homomorphisms we use must be able to map self-organising components onto each other. That is, $pe$-graph homomorphisms need to map cycles to cycles. This will allow the biochemist and the ecologist, for example, to compare their points of view with each other whilst maintaining a consistent analysis of the self-organisation that is occurring in the physical system under investigation. We now develop the mathematical tools required to conduct such a comparative analysis of $pe$-graphs, and begin to do so by introducing the following notions of \textit{preservation} and \textit{reflection of closure}.

\begin{definition}[Preservation of closure]
Let $G$ and $H$ be graphs and let $\phi \colon G \to H$ be a weak graph homomorphism. Then $\phi$ \textit{preserves closure} if for every cycle $C \subseteq G$ there exists a cycle $D \subseteq H[\phi(C)]$.
\label{def: presclosure}
\end{definition}

Figure \ref{fig: closure} illustrates the notion of preservation of closure.
\begin{figure}[ht]
    \centering 
    % A -> B
    \begin{tikzpicture}[node distance={0.5cm}]
        % Nodes
        \node[text=blue] (A) at (0,0)     {$A$};
        \node            (B) at (1,0)     {$B$};
        \node            (G) at (0.5,-1.33) {};
        % Arrows
        \draw[->] (A) -- (B);
        % Border
        \node [draw, fit=(A)(B)] {};
    \end{tikzpicture}
    % phi
    \begin{tikzpicture}[node distance={0.5cm}]
        % Nodes
        \node (A) at (0,0)     {};
        \node (B) at (1.2,0) {};
        \node (C) at (0,-1.3) {};
        \node (D) at (0, 1.3) {};
        % Arrows
        \draw[<-] (A) -- (B) node[midway, above] {$\phi$};
    \end{tikzpicture}
    % A1, A2, A3 cycle -> B
    \begin{tikzpicture}[node distance={0.5cm}]
        % Nodes
        \node[text=blue] (A1) at (0,0)     {$A_1$};
        \node[text=blue] (A2) at (2,1)     {$A_2$};
        \node[text=blue] (A3) at (0,2)     {$A_3$};
        \node            (B)  at (3.7,1)     {$B$};
         % Arrows
        \draw[->, blue] (A1) -- (A2);
        \draw[->, blue] (A2) -- (A3);
        \draw[->, blue] (A3) -- (A1);
        \draw[->] (A2) -- (B);
        % Border
        \node [draw, fit=(A1)(A2)(A3)(B)] {};
    \end{tikzpicture}
    % psi
    \begin{tikzpicture}[node distance={0.5cm}]
        % Nodes
        \node (A) at (0,0)     {};
        \node (B) at (1.2,0) {};
        \node (C) at (0,-1.3) {};
        \node (D) at (0, 1.3) {};
        % Arrows
        \draw[->] (A) -- (B) node[midway, above] {$\psi$};
    \end{tikzpicture}
    % A loop -> B
    \begin{tikzpicture}[node distance={0.5cm}]
        % Nodes
        \node[text=blue] (A) at (0,0)     {$A$};
        \node            (B) at (1,0)     {$B$};
        \node            (G) at (0.25,-1.3) {};
        % Arrows
        \draw[->] (A) -- (B);
        \draw[->, blue] (A.north) to [out=145,in=215, looseness=4] (A.south);
        % Border
        \draw (-0.84, 0.7) rectangle (1.5, -0.7);
    \end{tikzpicture}
    \caption{Both $\phi$ and $\psi$ are weak graph homomorphisms, but $\psi$ preserves closure whilst $\phi$ does not. Both $\phi$ and $\psi$ map $A_1 \mapsto A$, $A_2 \mapsto A$, and $A_3 \mapsto A$. The induced image of the cycle in the centre graph, under $\phi$ and $\psi$, is highlighted in blue in their respective target graphs.}
    \label{fig: closure}
\end{figure}
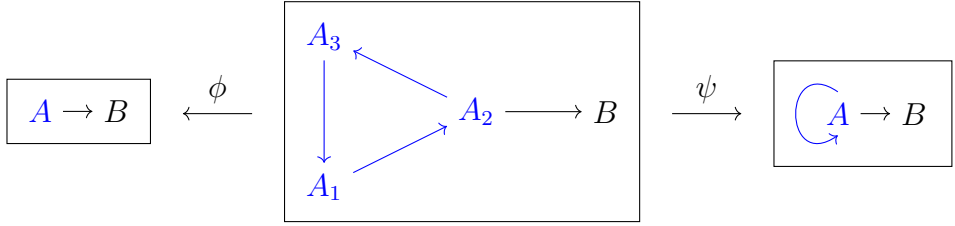

Note that standard graph homomorphisms always preserve closure, but this is not necessarily true for \textit{weak} graph homomorphisms. In Theorem \ref{thm: prestest} in the Appendix, we provide a useful test to check whether a weak graph homomorphism preserves closure. Since we want to be able to contract edges into vertices, but also want homomorphisms between $pe$-graphs to preserve closure, we define homomorphisms of $pe$-graphs as follows.

\begin{definition}[Homomorphism of $pe$-graphs]
Consider a system $S$, and let $G$ and $H$ be $pe$-graphs which model subsystems of $S$ over the same time interval. Let $\phi \colon G \to H$ be a weak graph homomorphism that preserves closure. Then $\phi$ is a \textit{homomorphism of $pe$-graphs} if for every process $p_{_X} \in V(G)$ we have that $p_{_X}$ is happening within $\phi(p_{_X})$. That is, $X$ is contained within $Y$ where $\phi(p_{_X}) = q_{_Y} \in V(H)$.
\end{definition}

We also refer to homomorphisms of $pe$-graphs as \textit{$pe$-graph homomorphisms}, or simply \textit{homomorphisms}, if it is sufficiently clear that we are referring to $pe$-graphs. We will often abbreviate the notation when referring to a $pe$-graph homomorphism and say `$p$ is happening within $\phi(p)$' when we really mean `$p_{_X}$ is happening within $q_{_Y}$' where $q_{_Y} = \phi(p_{_X})$. In Theorems \ref{thm: weakcomp} to \ref{thm: PEcomp} in the Appendix, we show that $pe$-graph homomorphisms compose.

Homomorphisms of $pe$-graphs allow us to translate between two perspectives representing the same system, mapping all of the relevant self-organisational properties of one $pe$-graph to another. However, it is possible for the target graph to contain \textit{more} cycles than the source graph. But if every cycle in the target graph is reflected appropriately in the source graph then we say that a homomorphism \textit{reflects closure}.

\begin{definition}[Reflection of closure]

Let $G$ and $H$ be graphs and let $\phi \colon G \to H$ be a weak graph homomorphism. Then $\phi$ \textit{reflects closure} if for every cycle $D \subseteq H$ there exists a cycle $C \subseteq G$ such that $D \subseteq H[\phi(C)]$.
\label{def: reflclosure}
\end{definition}

A homomorphism which reflects closure provides a correspondence between all of the key self-organisational features of the two $pe$-graphs, namely the processes and the cycles. As such, these maps indicate a greater similarity between two $pe$-graphs. We therefore call these maps \textit{homorheisms} to distinguish them from homomorphisms which do not necessarily reflect closure.

\begin{definition}[Homorheism]
Consider a system $S$, and let $G$ and $H$ be $pe$-graphs which model subsystems of $S$ over the same time interval. Let $\phi \colon G \to H$ be a homomorphism of $pe$-graphs. If $\phi$ also reflects closure then $\phi$ is a \textit{homorheism}. We will denote homorheisms using the modified arrow $\circlearrow$.
\end{definition}

We show in Theorems \ref{thm: reflcomp} to \ref{thm: homorheismcomp} in the Appendix that homorheisms also compose. In Theorem \ref{thm: refltest} in the Appendix, we also provide another simple test to check whether a weak graph homomorphism reflects closure. The notion of a homorheism is illustrated in Figure \ref{fig: reflect}. 
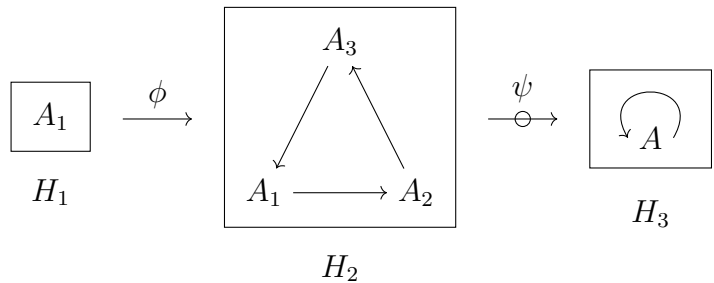
\begin{figure}[ht]
    \centering 
    % A_1
    \begin{tikzpicture}[node distance={0.5cm}]
        % Nodes
        \node (A1) at (0,0)     {$A_1$};
        \node (H) at (0,-1) {$H_1$};
        \node (phantom) at (0,-2.18) {};
        % Border
        \node [draw, fit=(A1)] {};
    \end{tikzpicture}
    % f
    \begin{tikzpicture}[node distance={0.5cm}]
        % Nodes
        \node (A) at (0,0)     {};
        \node (B) at (1.2,0) {};
        \node (C) at (0,-2.15) {};
        \node (D) at (0, 2.15) {};
        % Arrows
        \draw[->] (A) -- (B) node[midway, above] {$\phi$};
    \end{tikzpicture}
    % A1, A2, A3 cycle
    \begin{tikzpicture}[node distance={0.5cm}]
        % Nodes
        \node (A1) at (0,0)     {$A_1$};
        \node (A2) at (2,0)     {$A_2$};
        \node (A3) at (1,2)     {$A_3$};
        \node            (H)  at (1,-1)     {$H_2$};
         % Arrows
        \draw[->] (A1) -- (A2);
        \draw[->] (A2) -- (A3);
        \draw[->] (A3) -- (A1);
        % Border
        \node [draw, fit=(A1)(A2)(A3)] {};
    \end{tikzpicture}
    % g
    \begin{tikzpicture}[node distance={0.5cm}]
        % Nodes
        \node (A) at (0,0)     {};
        \node (B) at (1.2,0) {};
        \node (C) at (0,-2.15) {};
        \node (D) at (0, 2.15) {};
        % Arrows
        \draw[->] (A) -- (B) node[midway, above, inner sep={6pt}] {$\psi$};
        \draw (0.6, 0) circle [radius=0.1];
    \end{tikzpicture}
    % A loop -> B
    \begin{tikzpicture}[node distance={0.5cm}]
        % Nodes
        \node (A) at (0,0)     {$A$};
        \node (H) at (0,-1) {$H_3$};
        \node (phantom) at (0,-1.9) {};
        % Arrows
        \draw[->] (A.east) to [out=60,in=120, looseness=4] (A.west);
        % Border
        \draw (-0.8, 0.9) rectangle (0.8, -0.4);
    \end{tikzpicture}
    \caption{$H_1$, $H_2$, and $H_3$ are $pe$-graphs. Since $pe$-graph homomorphisms preserve closure, they exist only from $H_1 \to H_2$, $H_1 \to H_3$, and $H_2 \to H_3$, but not from $H_2 \to H_1$, $H_3 \to H_2$, and $H_3 \to H_1$. Of the three $pe$-graph homomorphisms only $\psi \colon H_2 \to H_3$, where $\psi$ maps each vertex in $H_2$ to $A$ in $H_3$, reflects closure and it is therefore a homorheism.}
    \label{fig: reflect}
\end{figure}

Homorheisms allow us to identify two vantage points that convey roughly the same description of the system under investigation, with respect to their self-organisational features. Explicitly, what we mean by this is that every cycle in the source graph is captured by a cycle in the target graph and vice versa. Note, however, that the source graph may be a more finely-grained perspective than the target graph. As such, the source graph may contain far more cycles than the target graph and therefore homorheisms are not an equivalence relation on $pe$-graphs. That is, just because there exists a homorheism $G \circlearrow H$ does not necessarily mean there is a homorheism $H \circlearrow G$.

But if there are \textit{no} homorheisms between two $pe$-graphs that represent the same system, then there is at least one fundamentally different organisational feature between them. Namely one graph must contain a cycle that the other does not capture in any way. Graphs like these are therefore different vantage points that offer significantly different perspectives of the same system, with respect to their organisational structure. It is possible, however, to use intermediate $pe$-graphs to translate between perspectives when there are no homorheisms directly between the graphs available. We will demonstrate this process in Section \ref{sect: FA}. 

As a linguistic note to conclude the discussion here, the term homorheism translates to \textit{same flow}. It is inspired by C. H. Waddington's neologism \textit{homeorhesis} to describe dynamical systems that return to the same dynamical trajectory, following a disturbance \cite{waddington1957strategy}.

\subsection{Self-enablement}

Our formalism allows for a $pe$-graph to have loops, but we have not yet described what this means conceptually. We wish to avoid paradoxes of self-enablement, so we do not want a loop in a $pe$-graph to indicate that a process is directly enabling itself. Instead, we use loops as a convenient shorthand. Consider a system $S$ of processes and direct enablements (the prokaryotic cell is a good example). If this system $S$ contains at least one cycle then instead of writing all of the processes and enablements within $S$ we can simply write $S \to S$. Therefore, a loop within a $pe$-graph should always indicate to the reader that the loop in question is really a finer-grained network of processes containing at least one cycle. There are, therefore, really two kinds of arrows in $pe$-graphs: direct enablements (not loops) and loops.

In contrast, many models often choose to ignore lower-level organisation and treat systems as if they were not self-organising. For instance, organisms within many evolutionary models are treated as objects that can be mostly reduced down to their genes \cite{lewontin1983organism, walsh2018objectcy, chiu2022extended}. For this reason, population genetics has been referred to pejoratively as `beanbag genetics' \cite{rao2011jbs}. We can see this disappearance of the self-organising organism in, for example, the Price equation where the self-organising properties of individuals are ignored in favour of a simple statistical equation \cite{price1970selection, frank2012natural}. Accordingly, if we were to treat such models as reality and then constructed the resulting $pe$-graphs, loops like `Organism $\to$ Organism' would never appear. 

With the use of loops as a shorthand we need to keep in mind that more detailed descriptions of a system may result in the loss of a homorheism. For instance, suppose $G$ is the graph with a single process `Hepatocyte' and a single loop (since cells are self-organising), and $H$ is a similar graph with a single process `Liver' and a single loop. Then there is a natural homorheism $\phi \colon G \circlearrow H$, but clearly $G$ and $H$ model quite different systems. If we unpack the loops in $G$ and $H$ and express these graphs as more detailed networks of the systems in question then we lose the homorheism: once we unpack the process `Liver' we will see how intricate and complex the network here really is, with hundreds of processes and cycles between many different cell types. Indeed, there are far more process-enablement cycles within the liver compared to those just occurring within a single hepatocyte.

\section{Applications} \label{sect: applications}

We have now developed our mathematical framework sufficiently to begin comparing a wide array of perspectives from across the biosphere. In practice, our analysis in the following subsections proceeds by first identifying the processes and enablements within a system, from a particular perspective, and then drawing the corresponding $pe$-graph. Since multiple perspectives are always possible there are bound to be local differences between these graphs. For instance, we may view two processes $p$ and $q$ where $p$ enables $q$ as distinct, so $p \to q$, or we may see $p$ and $q$ as a single compound process. Homomorphisms and homorheisms will allow us to compare these differences precisely and concretely.

\subsection{Autopoiesis} \label{sect: autop}

Humberto Maturana and Francisco Varela together developed the concept of \textit{autopoiesis}, which proved to be a hugely influential account of life \cite{maturana1980autopoiesis}. The concept went through a number of iterations, but we will focus on a later definition given by \citeA[p. 34]{varela2000fenomeno} and discussed in relation to cognition in \citeA[pp. 99--103]{thompson2007mind}. Here, a system is defined as being \textit{autopoietic} if the following criteria all hold:
    \begin{enumerate}
        \item The system has a semipermeable boundary constructed from molecular components. This boundary allows us to discriminate between the inside and outside of the system. \label{auto1}
        \item The components of the boundary are produced by a metabolic reaction network (MRN) located within the boundary. \label{auto2}
        \item There is an interdependence between Criteria \ref{auto1} and \ref{auto2}. That is, the MRN is regenerated and maintained because of the existence of the boundary which keeps all of the molecules in the system in close proximity to each other. \label{auto3}
    \end{enumerate}

For an autopoietic system to be maintained, much like an autocatalytic set, there must also be a ready supply of external nutrients that can flow into the system. Although this final criterion is not made explicit by Maturana and Varela, it is implicitly assumed throughout their work \cite{varela2000fenomeno}. The prototypical example of an autopoietic system is the prokaryotic cell. It is still debated whether multicellular organisms, eusocial insect colonies, or even human societies are also autopoietic. The discussion around these trickier cases typically centres on how one defines a semipermeable boundary \cite{thompson2007mind}. 

From this basic outline we can see that a minimal autopoietic system has three key processes: the MRN, boundary maintenance, and the passive diffusion of food into the system. Food diffusion directly enables the MRN since it provides the basic reactants for the system, and boundary maintenance and the MRN directly enable each other because of Criteria \ref{auto2} and \ref{auto3} (Figure \ref{fig: autop}, $A_1$).
\begin{figure}[ht]
    \centering 
    \begin{tikzpicture}[node distance={0.5cm}]
        % Nodes
        \node[align=center] (FD) at  (0,0)     {\footnotesize Food\\ \footnotesize Diffusion};
        \node[align=center] (MRN) at (1.1,2)     {\footnotesize MRN};
        \node[align=center] (BM) at  (2.2,0)     {\footnotesize Boundary\\ \footnotesize Maintenance};
        \node[align=center] (A1) at  (1,-1.4)    {$A_1$};
        % Arrows
        \draw[->] (FD) -- (MRN);
        \draw[transform canvas={xshift=1.5ex}, ->] (MRN) --  (BM);
        \draw[transform canvas={xshift=0.3ex}, ->] (BM) --  (MRN);
        % Border
        \draw (-1, -0.9) rectangle (3.45, 2.7);
    \end{tikzpicture}
    \hspace{0.1cm}
    \begin{tikzpicture}[node distance={0.5cm}]
        % Nodes
        \node[align=center] (FD) at  (0,0)     {\footnotesize Food\\ \footnotesize Diffusion};
        \node[align=center] (MRN) at (1.1,2)     {\footnotesize MRN};
        \node[align=center] (BM) at  (2.2,0)     {\footnotesize Boundary\\ \footnotesize Maintenance};
        \node[align=center] (A1) at  (1,-1.2)    {$A_2$};
        % Arrows
        \draw[->] (FD) -- (MRN);
        \draw[transform canvas={xshift=1.5ex}, ->] (MRN) --  (BM);
        \draw[transform canvas={xshift=0.3ex}, ->] (BM) --  (MRN);
        \draw[->] (MRN) to [out=55,in=125, looseness=5] (MRN);
        % Border
        \draw (-1, -0.7) rectangle (3.45, 2.9);
    \end{tikzpicture}
    \hspace{0.1cm}
    \begin{tikzpicture}[node distance={0.5cm}]
        % Nodes
        \node[align=center] (T) at  (0,0)     {\footnotesize Transport};
        \node[align=center] (MRN) at (1.25,2)     {\footnotesize MRN};
        \node[align=center] (BM) at  (2.5,0)     {\footnotesize Boundary\\ \footnotesize Maintenance};
        \node[align=center] (A1) at  (1.25,-1.2)    {$A_3$};
        % Arrows
        \draw[transform canvas={xshift=-1.5ex}, ->] (T) --  (MRN);
        \draw[transform canvas={xshift=-0.3ex}, ->] (MRN) --  (T);
        \draw[transform canvas={xshift=1.5ex}, ->] (MRN) --  (BM);
        \draw[transform canvas={xshift=0.3ex}, ->] (BM) --  (MRN);
        \draw[->] (MRN) to [out=55,in=125, looseness=5] (MRN);
        \draw[->] (BM) --  (T);
        % Border
        \draw (-1, -0.7) rectangle (3.72, 2.9);
    \end{tikzpicture}
    \caption{Three different kinds of autopoietic systems represented as $pe$-graphs. MRN = metabolic reaction network. $A_1$: minimal autopoiesis. $A_2$: autopoiesis with an organisationally closed metabolic reaction network. $A_3$: autopoiesis with facilitated diffusion, active transport, and endocytosis of relevant molecules.}
    \label{fig: autop}
\end{figure}
There is no direct requirement that the MRN itself be organisationally closed, but in practice life has incorporated many cycles into its metabolic reaction networks (for example, the Krebs cycle). For these more complex autopoietic systems, the MRN itself achieves organisational closure (Figure \ref{fig: autop}, $A_2$). Some autopoietic systems also play an active role in bringing nutrients and other molecules into the system. For instance, cells often use processes like facilitated diffusion, primary and secondary active transport, and endocytosis to control which molecules enter into their internal metabolic reaction networks \cite{jeckelmann2020transporters}. In these autopoietic systems the MRN directly enables transport by producing the relevant transport proteins (for example GLUTs, Na$^+$/K$^+$-ATPase, and Clathrin), and the maintenance of the boundary directly enables these transport processes since they would not occur without a cell membrane (Figure \ref{fig: autop}, $A_3$).

So using the language of $pe$-graphs we can see at least three distinct kinds of autopoietic systems, with differing levels of organisation (Figure \ref{fig: autop}). In isolation, the $pe$-graphs above are not very interesting. In order to see how and why $pe$-graphs are a useful tool for biologists, we need to introduce a second, complementary, perspective of the cell which we can then compare to autopoiesis.

\subsection{(F,A)-systems}\label{sect: FA}

Robert Rosen used the language of category theory to develop a highly abstract theory of biological organisation which he called $(M,R)$-systems. Rosen's framework was not only mathematically intricate, but it also contained essentially no reference to biochemistry. As a result, Rosen's work has been largely ignored by biologists. However, his ideas have been recently refined by \citeA{hofmeyr2019basic, hofmeyr2021biochemically} to map onto real biochemical processes. The result is what Hofmeyr calls fabrication-assembly $(F,A)$-systems. As considered within the context of the cell, an $(F,A)$-system has three classes of processes: (1) covalent metabolic chemistry (fabrication); (2) supramolecular chemical processes (assembly) and; (3) maintenance of the intracellular milieu via the transporting of ions.

Covalent molecular chemistry or \textit{fabrication} involves the transport of nutrients into the cell, the degradation of nutrients into the building blocks of metabolism (for example amino acids, lipids, and nucleotides) and the synthesis of macromolecules like polypeptides and nucleic acids from these building blocks. Supramolecular chemistry or \textit{assembly} involves the folding of proteins, the self-assembly of ribosomes, and the maintenance of the cell membrane via noncovalent interactions. This kind of chemistry is facilitated by the regulated ionic composition, temperature, and pH of the intracellular milieu, and by chaperone proteins. Finally, the maintenance of the intracellular milieu is primarily driven by the transport of ions into and out of the cell via membrane-bound transporters.

The enablements between these three classes of processes within an $(F,A)$-system can be summarised via three cycles (Figure \ref{fig: intermediate}, $F_A$).
\begin{figure}[ht]
    \centering 
    % Ip
    \begin{tikzpicture}[node distance={0.5cm}]
        % Nodes
        \node[align=center, half fill={blue!20}{red!25}] (NT) at (0,2.5)     {\footnotesize Nutrient \\ \footnotesize Transport};
        \node[align=center, half fill ={orange!30}{violet!30}] (MM) at (1.5,0)                  {\footnotesize Membrane\\ \footnotesize Maintenance};
        \node[align=center, half fill={blue!20}{green!20}] (CC) at (3,5)                  {\footnotesize Covalent\\ \footnotesize Chemistry};
        \node[align=center, half fill={orange!30}{green!20}] (PF) at (6,2.5)                  {\footnotesize Protein\\ \footnotesize Folding};
        \node[align=center, half fill={gray!25}{red!25}] (IT) at (4.5,0)                  {\footnotesize Ion\\ \footnotesize Transport};
        \node (I) at (3,6.3)                {$I_\mathcal{P}$};
        \node (phantom) at (0,-1) {};
        % Arrows
        \draw[->, shorten <= 0.2cm, shorten >= 0.2cm]  (NT) -- (CC);
        \draw[transform canvas={xshift=1.2ex}, ->, shorten <= 0.2cm, shorten >= 0.2cm]  (MM) -- (CC);
        \draw[->, shorten <= 0.3cm, shorten >= 0.2cm]  (MM) -- (PF);
        \draw[->, shorten <= 0.12cm, shorten >= 0.4cm, transform canvas = {yshift=-1.2ex}]  (PF) -- (MM);
        \draw[->, shorten <= 0.2cm, shorten >= 0.2cm]  (MM) -- (NT);
        \draw[->, shorten <= 0.2cm, shorten >= 0.2cm]  (MM) -- (IT);
        \draw[transform canvas={xshift=1.2ex}, ->, shorten <= 0.2cm, shorten >= 0.2cm]  (IT) -- (PF);
        \draw[transform canvas={xshift=0.3ex}, ->, shorten <= 0.2cm, shorten >= 0.2cm]  (PF) -- (CC);
        \draw[->, shorten <= 0.2cm, shorten >= 0.2cm]  (PF) -- (NT);
        \draw[->, shorten <= 0.2cm, shorten >= 0.2cm]  (PF) -- (IT);
        \draw[->, shorten <= 0.2cm, shorten >= 0.2cm]  (CC) -- (MM);
        \draw[transform canvas={xshift=2ex}, ->, shorten <= 0.2cm, shorten >= 0.2cm]  (CC) -- (PF);
        % Border
        \node [draw, fit=(IT)(MM)(CC)(PF)(NT), inner sep = 0.4cm] {};
    \end{tikzpicture}
    \hspace{5cm}
    % phi 2
    \begin{tikzpicture}[node distance={0.5cm}]
                % Nodes
        \node (A) at (1.1,1.1)     {};
        \node (B) at (0,0) {};
        % Arrows
        \draw[->] (A) -- (B) node[midway, above left] {$\phi_{_5}$};
        \draw (0.55, 0.55) circle [radius=0.1];
    \end{tikzpicture}
    % center gap
    \hspace{2.2cm}
    % phi 3
    \begin{tikzpicture}[node distance={0.5cm}]
                % Nodes
        \node (A) at (0,0)     {};
        \node (B) at (1.1,-1.1) {};
        % Arrows
        \draw[->] (A) -- (B) node[midway, above right] {$\phi_{_6}$};
        \draw (0.55, -0.55) circle [radius=0.1];
    \end{tikzpicture}
    %middle row pad at end to prevent wrapping
    \hspace{10cm}
    % F
    \begin{tikzpicture}[node distance={0.5cm}]
        % Nodes
        \node[align=center, fill={blue!20}] (F) at (0,3.5)     {\footnotesize Fabrication};
        \node[align=center, fill={orange!30}] (A) at (0,2)                  {\footnotesize Assembly};
        \node[align=center, fill={gray!25}] (ICM) at (0,0)                  {\footnotesize Maintenance of\\ \footnotesize the intracellular \\ \footnotesize milieu};
        \node (FL) at (0,-1.5)                {$F_A$};
        % Arrows
        \draw[->, transform canvas={xshift=0.8ex}, shorten <= 0.2cm, shorten >= 0.2cm]  (F) -- (A);
        \draw[->, transform canvas={xshift=-0.8ex}, shorten <= 0.2cm, shorten >= 0.2cm]  (A) -- (F);
        \draw[->, transform canvas={xshift=0.8ex}, shorten <= 0.2cm, shorten >= 0.2cm]  (A) -- (ICM);
        \draw[->, transform canvas={xshift=-0.8ex}, shorten <= 0.2cm, shorten >= 0.2cm]  (ICM) -- (A);
        \draw[->, shorten <= 0.4cm, shorten >= 0.4cm] (A.north) to [out=145,in=215, looseness=7] (A.south);
        % Border
        \node [draw, fit=(F)(A)(ICM), inner sep=0.3cm] {};
    \end{tikzpicture}
    %bottom row center pad
    \hspace{1.4cm}
    % A3
    \begin{tikzpicture}[node distance={0.5cm}]
        % Nodes
        \node[align=center, fill={red!25}] (T) at  (0,0)     {\footnotesize Transport};
        \node[align=center, fill={green!20}] (MRN) at (1.7,2.5)     {\footnotesize MRN};
        \node[align=center, fill={violet!30}] (BM) at  (3.4,0)     {\footnotesize Boundary\\ \footnotesize Maintenance};
        \node[align=center] (A1) at  (1.7,-1.2)    {$A_3$};
        \node (phantom) at (0,-1.6) {};
        % Arrows
        \draw[transform canvas={xshift=-1.8ex}, ->, shorten <= 0.2cm, shorten >= 0.2cm] (T) --  (MRN);
        \draw[transform canvas={xshift=-0.3ex}, ->, shorten <= 0.2cm, shorten >= 0.2cm] (MRN) --  (T);
        \draw[transform canvas={xshift=1.8ex}, ->, shorten <= 0.2cm, shorten >= 0.2cm] (MRN) --  (BM);
        \draw[transform canvas={xshift=0.3ex}, ->, shorten <= 0.2cm, shorten >= 0.2cm] (BM) --  (MRN);
        \draw[->, shorten <= 0.2cm, shorten >= 0.2cm] (MRN) to [out=60,in=120, looseness=8] (MRN);
        \draw[->, shorten <= 0.2cm, shorten >= 0.2cm] (BM) --  (T);
        % Border
        \draw (-1.2, -0.8) rectangle (4.85, 3.7);
    \end{tikzpicture}
    \caption{A comparison of the $(F,A)$-system model of the cell to an autopoietic perspective using $pe$-graphs. $F_A$: $(F,A)$-system representation. $I_\mathcal{P}$: intermediate $pe$-graph. $A_3$: autopoietic representation. There exist homorheisms $\phi_5 \colon I_\mathcal{P} \circlearrow F_A$ and $\phi_6 \colon I_\mathcal{P} \circlearrow A_3$, which are described using a colouring to match vertices in $I_\mathcal{P}$ with those in $F_A$ or $A_3$ (see Figure \ref{fig: Ipcycles}). There are no homomorphisms of $pe$-graphs between $F_A$ and $A_3$, as explained in the main text. MRN = metabolic reaction network.}   
    \label{fig: intermediate}
\end{figure}
Fabrication enables assembly by providing the raw materials that need to be assembled. This enablement is direct because as macromolecules are being produced within the cell, they immediately begin to self-assemble as they interact with the intracellular milieu. Conversely, assembly enables fabrication by providing the enzymes needed to catalyse the covalent chemistry of the cell. This enablement is direct because enzymes are highly dynamic entities and constantly have to maintain their conformation by interacting with the intracellular milieu \cite{yang2003protein, nicholson2019cell}. As such, the assembly process is not just a one off event, but is an ongoing process, even as enzymes interact with their corresponding substrates. So we have the first cycle: Fabrication $\rightleftarrows$ Assembly.

In turn, the stochastic fluctuations of proteins are modulated by the constantly shifting, local composition of the intracellular milieu. Thus, the maintenance of the intracellular milieu directly enables assembly. Conversely, the maintenance of the cell membrane directly enables ions to be transported across it via membrane transporters, since without a membrane there would be no ion channels. Thus, assembly directly enables the maintenance of the intracellular milieu. This enablement is direct because the cell membrane is constantly shifting and undergoing repair. This gives us a second cycle: Assembly $\rightleftarrows$ Maintenance of the intracellular milieu.

Finally, assembly has a loop because protein folding and cell membrane maintenance, both processes occuring within assembly, directly enable each other. Maintenance of the cell membrane keeps ions within the cell, directly enabling protein folding. Conversely, the folding of enzymes like flippase and scramblase is needed to keep the membrane regulated and intact. 

Both the $(F,A)$-system and autopoiesis are models of the biological cell, and we can model each with the $pe$-graphs $F_A$ and $A_3$, respectively, as shown in Figure \ref{fig: intermediate}. But it is impossible to directly compare $F_A$ and $A_3$ via a $pe$-graph homomorphism. This is because the two perspectives partition the processes of the cell differently. Explicitly, assembly involves both the maintenance of the cell's boundary and protein folding; the latter is a part of the MRN as viewed from the perspective of autopoiesis. So if we wanted a $pe$-graph homomorphism $\phi \colon F_A \to A_3$, we would need to map assembly to both boundary maintenance and the MRN. Conversely, the MRN also involves the synthesis of nucleic acids from building blocks, which is a part of fabrication in $F_A$. So if there were a $pe$-graph homomorphism $\psi \colon A_3 \to F_A$, we would need to map the MRN onto both fabrication and assembly. Since a $pe$-graph homomorphism cannot map a single vertex to two vertices, there are no $pe$-graph homomorphisms between $F_A$ and $A_3$. Yet, both $F_A$ and $A_3$ model the same system, so we would expect some kind of connection between them. Indeed, we can find one, if we construct an intermediate perspective as follows.

First, we need to choose a finer partition $\mathcal{P}$ of the processes described in each of the models. There are many choices we could make here, and different sets of finer-grained processes will reveal different similarities between the models. As an illustrative example we will choose the set $\mathcal{P} :=$ \{nutrient transport, membrane maintenance, protein folding, ion transport, covalent chemistry\}. Here, covalent chemistry is equivalent to $(F,A)$-system fabrication, excluding nutrient transport, and is also equivalent to the autopoietic MRN, excluding protein folding. Membrane maintenance refers to the trafficking of lipids and proteins to the cell membrane, the lateral diffusion of membrane components, the flipping of lipids from one side of the membrane to the other, membrane repair processes, and the containment of cytosolic molecules within the cell as they bounce off the inner plasma membrane. Ion transport refers to the transport of ions into and out of the cell via facilitated diffusion and active transport.

Second, we match each of the processes in $\mathcal{P}$ to processes in the target graphs in question. We have illustrated this matching for our example in Figure \ref{fig: intermediate} using two colourings. We must ensure that the choice of $\mathcal{P}$ is fine enough such that every process $p \in \mathcal{P}$ is only happening within a \textit{single} process in each of the target graphs. This will ultimately allow us to construct well-defined homomorphisms.

Third, we construct the corresponding $pe$-graph, $I_\mathcal{P}$, by drawing in the direct enablements between processes in $\mathcal{P}$. If the arrows are drawn correctly then, using the correspondence from the previous step, we should have weak graph homomorphisms from $I_\mathcal{P}$ to the two target graphs. Finally, we need to check whether these weak graph homomorphisms preserve or reflect closure. If this is successful then we will have indeed constructed $pe$-graph homomorphisms, or possibly homorheisms, from the intermediate $pe$-graph to the two target graphs. We prove this for our case in Theorems \ref{thm: intermedpres} and \ref{thm: intermedrefl}. If, however, the weak graph homomorphisms do not preserve closure then we can choose a different set $\mathcal{P}$ and repeat the steps above. Alternatively, if we cannot construct homomorphisms from $I_\mathcal{P}$ to the target graphs then it may be an indication that there is a misplaced or missing arrow in these target graphs. Accordingly, this process of constructing an intermediate $pe$-graph can facilitate model development, as well as improve existing models. 

The existence of the graph $I_\mathcal{P}$ and corresponding homorheisms $\phi_5 \colon I_\mathcal{P} \to F_A$ and $\phi_6 \colon I_\mathcal{P} \to A_3$ in Figure \ref{fig: intermediate} gives some key insights. First, by providing a lower-level mechanistic model, the apparent discrepancy between $F_A$ and $A_3$ can be clarified as being different partitions of the same underlying groups of processes. The homorheisms $\phi_5$ and $\phi_6$ tell scientists exactly how these two models view the same set of processes and enablements differently, allowing for these two perspectives to be consistently and precisely compared. For example, from the perspective of autopoiesis, covalent chemistry and protein folding are both considered a part of the metabolic reaction network, but from the perspective of fabrication-assembly systems, the former is part of fabrication but the latter is part of assembly.

Second, the homorheisms $\phi_5$ and $\phi_6$ give two different explanations for the arrangement of process-enablement relations in $I_\mathcal{P}$. In isolation, it is not clear why the processes in $I_\mathcal{P}$ belong in the same $pe$-graph. Indeed, there were many choices available for the set of processes $\mathcal{P}$ and we chose the ones we did somewhat arbitrarily. However, this apparent arbitrariness is reduced via the homorheisms $\phi_5$ and $\phi_6$ as they collect several lower-level processes and enablements together into fewer higher-level processes and enablements. This brings disparate lower-level processes together and reveals higher-level connections between groups of processes. In an Aristotelian sense, the graphs $F_A$ and $A_3$ provide formal causes for the graph $I_\mathcal{P}$, realised through the corresponding homomorphisms. For example, ion transport and nutrient transport have no direct enablement between them in $I_\mathcal{P}$, but we can understand their co-existence in this graph because both are mapped via $\phi_6$ to the transport process in $A_3$, which is in turn enabled by the metabolic reaction network and boundary maintenance processes in $A_3$.

Third, and perhaps most crucially, the existence of the homorheisms $\phi_5$ and $\phi_6$ ensure that the \textit{organisational} features of $I_\mathcal{P}$ are captured in $F_A$ and $A_3$ respectively, not just the physical processes. This is a unique feature of the $pe$-graph framework. We could have certainly argued that autopoiesis and $(F,A)-$systems are compatible with the graph in $I_\mathcal{P}$ via several different means, but the graph theoretic nature of homorheisms ensures that the \textit{cycles} in each graph correspond to each other. This yields a highly precise confirmation that the self-organisation in one perspective is entirely captured in another. For example, the cycle Protein Folding $\rightleftarrows$ Membrane Maintenance in $I_\mathcal{P}$ maps onto Assembly $\circlearrowleft$ in $F_A$ and MRN $\rightleftarrows$ Boundary Maintenance in $A_3$.

The precise comparision facilitated by our formalism allows us to articulate and study the same self-organising processes from three different perspectives, enabling scientists to see exactly how each perspective understands the nature of their particular self-organising subsystems differently. Another important remark is that because $\phi_5$ and $\phi_6$ are homorheisms, not just homomorphisms, all of the cycles in $F_A$ (respectively $A_3$) can be found in $I_\mathcal{P}$ and are therefore embodied in $A_3$ (respectively $F_A$) via $\phi_6$ (respectively $\phi_5$). This means that we have not lost any of the self-organisational features when we move from $F_A$ to $A_3$ or vice versa, as they are just re-articulated in a different way and into different cycles. In short, we could simply say that $F_A$ and $A_3$ model the same self-organising processes in the cell, but they each have a different focus as to what they consider important, reflected in their cycle structures. In the Appendix, we consider the $pe$-graph $I_\mathcal{P}$ in more detail, showing that each of the edges in $I_\mathcal{P}$ is a direct enablement (Figure \ref{fig: Ipcycles}).

\subsection{Autocatalytic sets} \label{sect: autocat}

In this section, we wish to demonstrate the wide applicability of the $pe$-graph framework by applying it to a category of self-organising systems that are not strictly \textit{biological} systems -- autocatalytic sets. The concept of autocatalytic sets emerged during the 1980s \cite{dyson1982model, kauffman1986autocatalytic}. Dyson and Kauffman were both interested in the origins of life, and wanted to explore how random sets of chemical reactions could come together to form an organised network. Their ideas were then refined via a mathematical formalism known as RAF sets: \textbf{r}eflexively \textbf{a}utocatalytic systems generated by a \textbf{f}ood set \cite{hordijk2004detecting, hordijk2012structure}. Here, \textit{reflexive autocatalysis} means that every reaction in the system is catalysed by at least one molecule already in the system, and \textit{food-generated} means that all the molecules in the system can be ultimately generated from the food set through chains of reactions within the system (Figure \ref{fig: RAF}).

\begin{figure}[ht]
    \centering 
    \begin{tikzpicture}[node distance={0.5cm}]
        % Nodes
        \node (F1)   at (4,0.8)       {$F_1$};
        \node (F2)   at (1,0)       {$F_2$};
        \node (F3)   at (4,0)       {$F_3$};
        \node (F4)   at (7,0)       {$F_4$};
        \node (B1)   at (4,4.5)       {$B_1$};
        \node (B2)   at (3,1.8)       {$B_2$};
        \node (B3)   at (4.9,1.7)       {$B_3$};
        \node (R1)   at (1,2)       {Reaction 1};
        \node (R2)   at (7,2)       {Reaction 2};
        \node (R3)   at (4,3)       {Reaction 3};
        % Arrows
        \draw[->] (F3) -- (R1);
        \draw[->] (F2) -- (R1);
        \draw[->] (F3) -- (R2);
        \draw[->] (F4) -- (R2);
        \draw[->] (R3) -- (B2);
        \draw[->] (B3) -- (R3);
        \draw[->] (F1) -- (R3);
        \draw[->] (R2) -- (B3);
        \draw[->] (R1) -- (B1);
        \draw[->, dashed] (F1) -- (R1);
        \draw[->, dashed] (B1) -- (R3);
        \draw[->, dashed] (B1) -- (R2);
        % Border
        %\node [draw, fit=(A1)] {};
    \end{tikzpicture}
    \caption{A RAF set. The food set is $\mathcal{S} = \{F_1, F_2, F_3, F_4\}$, while $B_1$, $B_2$, and $B_3$ are intermediate molecules produced by the network. Dotted arrows indicate catalysis and solid arrows represent molecules entering/exiting a reaction. Modified from Figure 4 in Hordijk, Steel, and Kauffman (2012).}
    \label{fig: RAF}
\end{figure}

To study autocatalytic sets using the language of $pe$-graphs, we first need to consider the kinds of processes and enablement relationships present within these chemical systems. The processes in each case are relatively clear: they are the chemical reactions of the autocatalytic set, together with the process of food diffusing freely into the system. Since at least one food molecule is required for each reaction displayed in Figure \ref{fig: RAF}, food diffusion must directly enable each reaction. However, whether reactions directly enable each other depends on the nature of molecule degradation and the rate constants of the reactions. For instance, if molecule $B_3$ degrades rapidly within the system, then Reaction 2 will directly enable Reaction 3 because Reaction 2 will need to continuously produce $B_3$ in order to keep Reaction 3 going (Figure \ref{fig: RAF}). On the other hand, if $B_3$ is more stable and remains abundant within the system then this direct enablement could be undermined. For example, if $B_3$ has been stockpiled by the system in large quantities, then the occurrence of Reaction 2 may be irrelevant to Reaction 3. For now, we will assume that this RAF set operates in a high degradation environment, so that whenever a reaction ceases then its respective products disappear almost instantaneously from the system\footnote{We could also allow for a low-degradation environment, but add the diffusion processes of each molecule into the $pe$-graph. For simplicity we instead assume a high-degradation environment, yielding direct enablements between reactions and ignoring the diffusion processes.}. We can then draw two $pe$-graphs to describe the RAF sets: $R_1$ and $R_2$ (Figure \ref{fig: $pe$-RAF}). 

\begin{figure}[ht]
    \centering 
    \begin{tikzpicture}[node distance={0.5cm}]
        % Nodes
        \node[align=center] (F) at (2,0)     {\footnotesize Food \\ \footnotesize Diffusion};
        \node[align=center] (R1) at (0,2)                  {\footnotesize Reaction\\ \footnotesize 1};
        \node[align=center] (R2) at (2,4)                  {\footnotesize Reaction\\ \footnotesize 2};
        \node[align=center] (R3) at (4,2)                  {\footnotesize Reaction\\ \footnotesize 3};
        \node (R) at (2,-1.2)                {$R_1$};
        % Arrows
        \draw[->]  (F) -- (R2);
        \draw[->]  (F) -- (R1);
        \draw[->]  (F) -- (R3);
        \draw[->]  (R1) -- (R2);
        %\draw[->]  (R3) -- (R2);
        \draw[->]  (R2) -- (R3);
        \draw[->]  (R1) -- (R3);
        % Border
        \node [draw, fit=(F)(R1)(R2)(R3)] {};
    \end{tikzpicture}
    \begin{tikzpicture}[node distance={0.5cm}]
                % Nodes
        \node (A) at (0,0)     {};
        \node (B) at (1.4,0) {};
        \node (C) at (0,-3.3) {};
        \node (D) at (0, 3.3) {};
        % Arrows
        \draw[->] (A) -- (B) node[midway, above, inner sep={6pt}] {$\phi_{_1}$};
    \end{tikzpicture}
    \begin{tikzpicture}[node distance={0.5cm}]
        % Nodes
        \node[align=center] (F) at (2,0)     {\footnotesize Food \\ \footnotesize Diffusion};
        \node[align=center] (R1) at (0,2)    {\footnotesize Reactions\\ \footnotesize 1+3};
        \node[align=center] (R3) at (4,2)    {\footnotesize Reaction\\ \footnotesize 2};
        \node (R) at (2,-1.2)                {$R_2$};
        \node (phantom) at (0,-1.78) {};
        % Arrows
        \draw[->]  (F) -- (R1);
        \draw[->]  (F) -- (R3);
        \draw[->]  (R1.north) to [out=40,in=140] (R3.north);
        \draw[->]  (R3) -- (R1);
        % Border
        \draw (-1.2, -0.7) rectangle (5.1, 3.6);
    \end{tikzpicture}
    \caption{Two $pe$-graph representations of the RAF set from Figure \ref{fig: RAF}.}
    \label{fig: $pe$-RAF}
\end{figure}

Interestingly, the $pe$-graph $R_1$ does not have a cycle, whilst $R_2$ does (Figure \ref{fig: $pe$-RAF}). This demonstrates that organisational closure is partially perspective dependent. Unlike \citeA{cusimano2020objectivity}, we do not see this failure of objectivity as a disadvantage of the $pe$-graph framework. Rather, we see it is a significant advantage because it introduces a way to identify qualitatively different analyses of the system under investigation. A comparison of these perspectives can then inform future empirical investigations and help us to better understand the system. For instance, if we see Reactions 1 and 3 as really two parts of the broader process `Reactions $1 + 3$', then the analysis provided in $R_2$ will make more sense. The system contains a cycle and this part of the system therefore achieves organisational closure.

We can also connect this perspective of life to autocatalytic sets by identifying homomorphisms between the two accounts. To do this, we first need to consider a system that could be viewed either through the lens of autopoiesis or autocatalytic sets. For example, we can imagine an extended version of the RAF set in Figure \ref{fig: RAF} where molecule $B_2$ is a lipid which is produced in sufficient quantity to form a semipermeable boundary. This boundary then keeps all of the relevant molecules of the RAF set together, and is completely permeable to the food molecules $F_1$, $F_2$, $F_3$, and $F_4$. Depending on how we view the RAF set as discussed in the previous section, and depending on how much of the system we would like to study, this analysis yields four related $pe$-graphs (Figure \ref{fig: RAF-autop}). 

Again, the homomorphisms $\phi_2$ and $\phi_3$ in Figure \ref{fig: RAF-autop} allow us to translate between perspectives and give us more information about each model. On the one hand, the $pe$-graphs $R_1$ and $R_2$ unravel the metabolic reaction network processes in $A_1$ and $A_2$ and provide a more detailed mechanistic explanation for how the metabolic reaction is realised. But on the other hand, the $pe$-graphs $A_1$ and $A_2$ provide an explanation for why the reaction networks in $R_1$ and $R_2$ happen to exist in the first place: they are enabled via the boundary maintenance processes in $A_1$ and $A_2$.

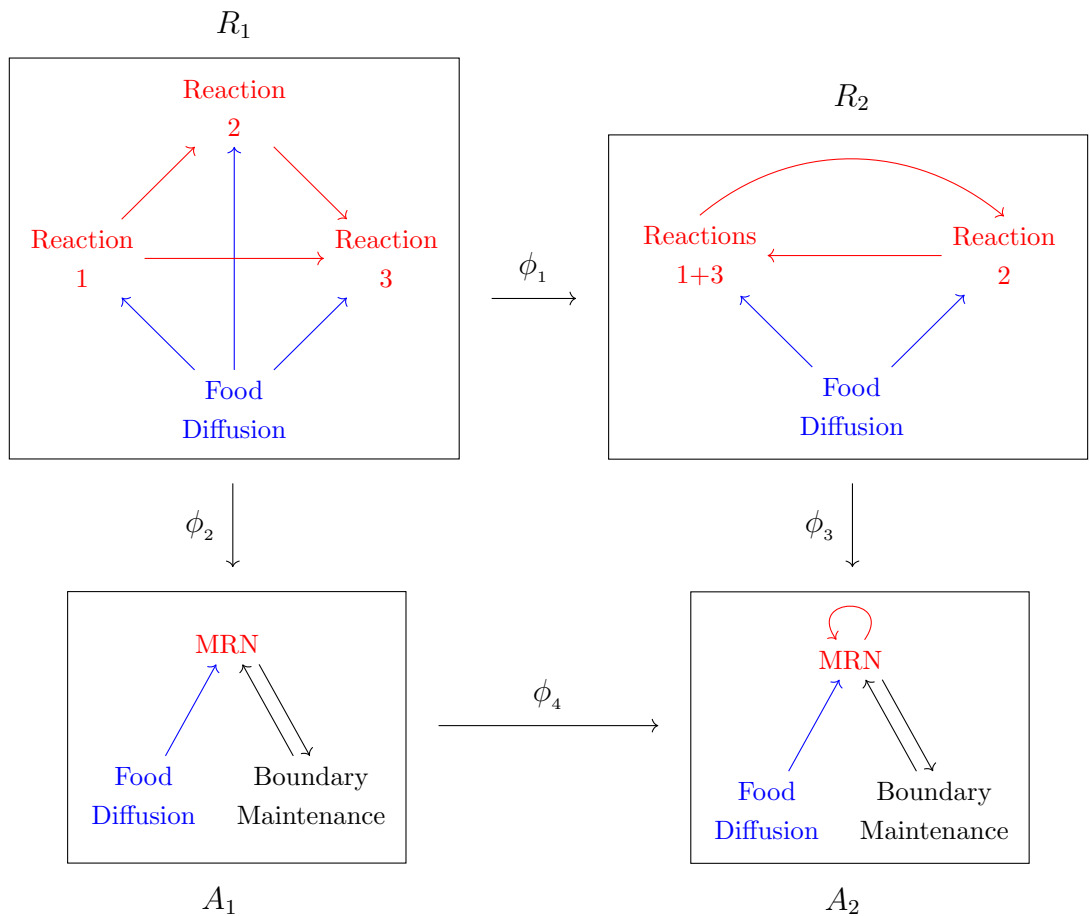
\begin{figure}[ht]
    \centering 
    % R1
    \begin{tikzpicture}[node distance={0.5cm}]
        % Nodes
        \node[align=center, text=blue] (F) at (2,0)     {\footnotesize Food \\ \footnotesize Diffusion};
        \node[align=center, text=red] (R1) at (0,2)                  {\footnotesize Reaction\\ \footnotesize 1};
        \node[align=center, text=red] (R2) at (2,4)                  {\footnotesize Reaction\\ \footnotesize 2};
        \node[align=center, text=red] (R3) at (4,2)                  {\footnotesize Reaction\\ \footnotesize 3};
        \node (R) at (2,5.1)                {$R_1$};
        % Arrows
        \draw[->, blue]  (F) -- (R2);
        \draw[->, blue]  (F) -- (R1);
        \draw[->, blue]  (F) -- (R3);
        \draw[->, red]  (R1) -- (R2);
        %\draw[->, red]  (R3) -- (R2);
        \draw[<-, red]  (R3) -- (R2);
        \draw[->, red]  (R1) -- (R3);
        % Border
        \node [draw, fit=(F)(R1)(R2)(R3)] {};
    \end{tikzpicture}
    % phi_1
    \begin{tikzpicture}[node distance={0.5cm}]
                % Nodes
        \node (A) at (0,0)     {};
        \node (B) at (1.4,0) {};
        \node (C) at (0,-2) {};
        \node (D) at (0, 2) {};
        % Arrows
        \draw[->] (A) -- (B) node[midway, above, inner sep={6pt}] {$\phi_{_1}$};
    \end{tikzpicture}
    % R2
    \begin{tikzpicture}[node distance={0.5cm}]
        % Nodes
        \node[align=center, text=blue] (F) at (2,0)     {\footnotesize Food \\ \footnotesize Diffusion};
        \node[align=center, text=red] (R1) at (0,2)    {\footnotesize Reactions\\ \footnotesize 1+3};
        \node[align=center, text=red] (R2) at (4,2)    {\footnotesize Reaction\\ \footnotesize 2};
        \node (R) at (2,4.1)                {$R_2$};
        % Arrows
        \draw[->, blue]  (F) -- (R1);
        \draw[->, blue]  (F) -- (R2);
        \draw[->, red]  (R1.north) to [out=40,in=140] (R2.north);
        \draw[->, red]  (R2) -- (R1);

        % Border
        \draw (-1.2, -0.7) rectangle (5.1, 3.6);
    \end{tikzpicture}
    % phi 2
    \begin{tikzpicture}[node distance={0.5cm}]
                % Nodes
        \node (A) at (0,0)     {};
        \node (B) at (0,-1.4) {};
        \node (C) at (-2.5,0) {};
        \node (D) at (2.5, 0) {};
        % Arrows
        \draw[->] (A) -- (B) node[midway, left, inner sep={6pt}] {$\phi_{_2}$};
    \end{tikzpicture}
    % center gap
    \hspace{3.1cm}
    % phi 3
    \begin{tikzpicture}[node distance={0.5cm}]
                % Nodes
        \node (A) at (0,0)     {};
        \node (B) at (0,-1.4) {};
        \node (C) at (-2,0) {};
        \node (D) at (2.65, 0) {};
        % Arrows
        \draw[->] (A) -- (B) node[midway, left, inner sep={6pt}] {$\phi_{_3}$};
    \end{tikzpicture}
    % A1
    \begin{tikzpicture}[node distance={0.5cm}]
        % Nodes
        \node[align=center, text=blue] (FD) at  (0,0)     {\footnotesize Food\\ \footnotesize Diffusion};
        \node[align=center, text=red] (MRN) at (1.1,2)     {\footnotesize MRN};
        \node[align=center] (BM) at  (2.2,0)     {\footnotesize Boundary\\ \footnotesize Maintenance};
        \node[align=center] (A1) at  (1,-1.4)    {$A_1$};
        % Arrows
        \draw[->, blue] (FD) -- (MRN);
        \draw[transform canvas={xshift=1.5ex}, ->] (MRN) --  (BM);
        \draw[transform canvas={xshift=0.3ex}, ->] (BM) --  (MRN);
        % Border
        \draw (-1, -0.9) rectangle (3.45, 2.7);
    \end{tikzpicture}
    % phi 4
    \begin{tikzpicture}[node distance={0.5cm}]
                % Nodes
        \node (A) at (0,0)     {};
        \node (B) at (3.18,0) {};
        \node (C) at (0,-2.5) {};
        % Arrows
        \draw[->] (A) -- (B) node[midway, above, inner sep={6pt}] {$\phi_{_4}$};
    \end{tikzpicture}
    % A2
    \begin{tikzpicture}[node distance={0.5cm}]
        % Nodes
        \node[align=center, text=blue] (FD) at  (0,0)     {\footnotesize Food\\ \footnotesize Diffusion};
        \node[align=center, text=red] (MRN) at (1.1,2)     {\footnotesize MRN};
        \node[align=center] (BM) at  (2.2,0)     {\footnotesize Boundary\\ \footnotesize Maintenance};
        \node[align=center] (A1) at  (1,-1.2)    {$A_2$};
        % Arrows
        \draw[->, blue] (FD) -- (MRN);
        \draw[transform canvas={xshift=1.5ex}, ->] (MRN) --  (BM);
        \draw[transform canvas={xshift=0.3ex}, ->] (BM) --  (MRN);
        \draw[->, red] (MRN) to [out=55,in=125, looseness=5] (MRN);
        % Border
        \draw (-1, -0.7) rectangle (3.45, 2.9);
    \end{tikzpicture}
    \caption{Process-enablement graph homomorphisms between autocatalytic set representations and autopoietic representations. Each arrow between graphs is a homomorphism and $\phi_2$ and $\phi_3$ are homorheisms onto their induced images. The map $\phi_2$ (respectively $\phi_3$) sends each of the reactions in $R_1$ (respectively $R_2$) to the MRN in $A_1$ (respectively $A_2$). MRN = metabolic reaction network.}
    \label{fig: RAF-autop}
\end{figure}

\section{Conclusions}

Over the past century, many different authors have come to the same conclusion that biological systems have to continuously reproduce the conditions of their own existence in order to stay alive \cite{cornish2020contrasting}. More recently, biological self-organisation and self-determination have been formalised through the concept of \textit{organisational closure} \cite{moreno2015biological, montevil2015biological}. In this paper we have developed a tool, the process-enablement graph or simply $pe$-graph, which not only allows us to identify organisational closure, but also allows us to compare how different scientific perspectives understand the nature of self-organisation.

We discern organisationally closed components of a system as cycles within $pe$-graphs, and use homomorphisms of $pe$-graphs to map cycles to cycles. These homomorphisms allow us to precisely and concretely examine how self-organisation is articulated under different perspectives, and yields direct correspondences between the relevant cycles.

We use $pe$-graphs to conduct an analysis of autopoiesis, autocatalytic sets and $(F, A)$-systems \cite{maturana1980autopoiesis, kauffman1986autocatalytic, hofmeyr2021biochemically}. We demonstrate that autopoiesis and $(F,A)$-systems describe the same self-organising phenomena, but from distinctly different perspectives. The $pe$-graph framework allows us to see exactly how the self-organising subsystems are articulated differently by each model. In addition, we show that whether autocatalytic sets realise organisational closure or not is perspective dependent, and is not entirely objective \cite{cusimano2020objectivity}.

One limitation of the framework presented in this paper is that $pe$-graphs are static. In order to study the evolution of physical systems, we would need to be able to describe how $pe$-graphs can qualitatively change. This would facilitate the application of $pe$-graphs to fields like developmental and evolutionary biology where systems, and the processes that constitute them, change over time. Indeed, by using our framework to investigate how and when organisational closure is realised in a sequence of $pe$-graphs we could even study the emergence of life from abiotic systems. We intend to develop our framework in this direction in future work.

As a final point, we would like to highlight that our framework of $pe$-graphs and the homomorphisms between them is a kind of \textit{universal biology} \cite{mariscal2018we}. To do universal biology is to study life as it \textit{must} be, rather than studying life as it exists in this or that system. In this light we argue that all life \textit{must} be self-organising \cite{jaegernaturalizing, rosen1991life}. Nonetheless, our approach is not exclusively focused on living matter, nor does it give a clear demarcation between biotic and abiotic systems. This is a significant advantage as it can help us to understand `life's edge' and to study the transitions between life and non-life \cite{smith1997major}. That is, our framework allows us to study how self-organisation may arise in \textit{any} system. Thus, $pe$-graphs are a powerful tool both to advance our theories of life and to better understand self-organising systems.

\section*{Acknowledgements}

EB would like to thank L{\'{e}}o Diaz (University of Melbourne) and Adriana Zanca (University of Melbourne) for insightful discussions and comments. We would also like to thank the two anonymous reviewers for comments that helped us to significantly improve this manuscript.

\bibliographystyle{apacite}
\bibliography{bibliography.bib}

\appendix
\section{Appendix}
\renewcommand\thefigure{\thesection.\arabic{figure}}
\renewcommand\thetable{\thesection.\arabic{table}}
\setcounter{figure}{0}

\subsection{Classification of closed $pe$-graphs}

In the main text, we give a classification of \textit{strictly} closed $pe$-graphs. Here, we give a more general classification of closed $pe$-graphs as a whole, providing further evidence for the importance of cycles to studying organisational closure.

\begin{theorem}[Classification of closed $pe$-graphs]
    A $pe$-graph $G$ is closed if and only if, for each $v \in V(G)$, either $v$ is contained in a cycle, or $v$ is contained in a directed path with initial and final vertices contained in distinct cycles.
    \label{thm: characterisationclosed}
\end{theorem}
\begin{proof}
    Suppose $G$ is closed and $v \in V(G)$ is not contained in a cycle. Let $W$ be a directed path containing $v$, such that no vertex in $W$ is contained in a cycle, and $W$ has maximum length. Denote by $u$ and $w$ the initial and final vertices of $W$, respectively. Since $G$ is closed, there exist vertices $x$, $y \in V(G)$ such that $x u$, $wy \in E(G)$. Further, $x$ and $y$ are contained in cycles, since $W$ has maximum length. Since $v$ is not contained within a cycle, these cycles must be distinct. Therefore, $v$ is contained in a directed path with initial and final vertices in distinct cycles.

    Conversely, suppose that, for each $v \in V(G)$, either $v$ is contained in a cycle, or $v$ is contained in a directed path with initial and final vertices contained in distinct cycles. Then, there must exist $u$, $w \in V(G)$ such that $u \to v \to w$. So $G$ is closed.
    \end{proof}

\subsection{Properties of homomorphisms and homorheisms}

Here we prove some useful results which allow us to more readily use homomorphisms and homorheisms in practice. 

\subsubsection{Composition of homomorphisms and homorheisms}

First, we show that weak graph homomorphisms compose.

\begin{theorem}
    Let $G_1$, $G_2$, and $G_3$ be graphs, and let $\phi \colon G_1 \to G_2$ and $\psi \colon G_2 \to G_3$ be weak graph homomorphisms. Then $\psi \circ \phi \colon G_1 \to G_3$ is a weak graph homomorphism.
    \label{thm: weakcomp}
\end{theorem}
\begin{proof}
Let $uv \in E(G_1)$ and assume $\psi(\phi(u)) \neq \psi(\phi(v))$. Then $\phi(u) \neq \phi(v)$ so, since $\phi$ is a weak graph homomorphism, $\phi(u)\phi(v) \in E(G_2)$. Further, since $\psi$ is a weak graph homomorphism, $\psi(\phi(u))\psi(\phi(v)) \in E(G_3)$. Therefore, $\psi \circ \phi$ is a weak graph homomorphism.
\end{proof}

Next, we show that weak graph homomorphisms that preserve closure are composable.

\begin{theorem}
    Let $G_1, G_2$, and $G_3$ be graphs and $\phi \colon G_1 \to G_2$ and $\psi \colon G_2 \to G_3$ be weak graph homomorphisms that preserve closure. Then $\psi \circ \phi \colon G_1 \to G_3$ also preserves closure.
    \label{thm: prescomp}
\end{theorem}
\begin{proof}
     Suppose $C_1 \subseteq G_1$ is a cycle. Since $\phi$ preserves closure there exists a cycle $C_2 \subseteq G_2[\phi(C_1)]$, and since $\psi$ preserves closure there exists a cycle $C_3 \subseteq G_3[\psi(C_2)]$. It suffices to show $C_3 \subseteq G_3[(\psi \circ \phi)(C_1)]$. For this, note that $V(\psi(C_2)) \subseteq V((\psi \circ \phi)(C_1))$, so $C_3 \subseteq G_3[\psi(C_2)] \subseteq G_3[(\psi \circ \phi)(C_1)]$.
\end{proof}

By applying Theorem \ref{thm: prescomp}, we can now demonstrate that $pe$-graph homomorphisms are composable.

\begin{theorem}
    Consider a system $S$. Let $G_1, G_2$, and $G_3$ be $pe$-graphs representing subsystems of $S$ over the same time interval, and let $\phi \colon G_1 \to G_2$ and $\psi \colon G_2 \to G_3$ be $pe$-graph homomorphisms. Then $\psi \circ \phi \colon G_1 \to G_3$ is also a $pe$-graph homomorphism.
    \label{thm: PEcomp}
\end{theorem}
\begin{proof}  
    By Theorem \ref{thm: prescomp}, $\psi \circ \phi$ preserves closure. If $p_{_X}$ is a process in $S$ with $p_{_X} \in V(G_1)$ then, since $\phi$ is a $pe$-graph homomorphism, there exists a process $q_{_Y}$ in $S$ such that $\phi(p_{_X}) = q_{_Y} \in V(G_2)$ and $p_{_X}$ happens within $q_{_Y}$. Further, since $\psi$ is a $pe$-graph homomorphism, there exists a process $r_{_Z}$ in $S$ such that $\psi(\phi(p_{_X})) = \psi(q_{_Y}) = r_{_Z} \in V(G_3)$ and $\phi(p_{_X}) = q_{_Y}$ happens within $r_{_Z}$. It follows that $p_{_X}$ happens within $\psi(\phi(p_{_X}))$, hence $\psi \circ \phi$ is a $pe$-graph homomorphism.
\end{proof}

The next result demonstrates that weak graph homomorphisms that reflect closure are composable. 

\begin{theorem}
    Let $G_1, G_2$, and $G_3$ be graphs, and let $\phi \colon G_1 \to G_2$ and $\psi \colon G_2 \to G_3$ be weak graph homomorphisms that reflect closure. Then $\psi \circ \phi \colon G_1 \to G_3$ reflects closure.
    \label{thm: reflcomp}
\end{theorem}

\begin{proof}
    If $C_3 \subseteq G_3$ is a cycle then, since $\psi$ reflects closure, there exists a cycle $C_2 \subseteq G_2$ such that $C_3 \subseteq G_3[\psi(C_2)]$. Further, since $\phi$ reflects closure there exists a cycle $C_1 \subseteq G_1$ such that $C_2 \subseteq G_2[\phi(C_1)]$. It suffices to show $C_3 \subseteq G_3[(\psi \circ \phi)(C_1)]$. For this, note that $V(\psi(C_2)) \subseteq V((\psi \circ \phi)(C_1))$, so $C_3 \subseteq G_3[\psi(C_2)] \subseteq G_3[(\psi \circ \phi)(C_1)]$. Hence $\psi \circ \phi$ reflects closure.
\end{proof}

It follows from Theorems \ref{thm: PEcomp} and \ref{thm: reflcomp} that homorheisms are composable. However, for completeness we state this as the following result (Theorem \ref{thm: homorheismcomp}).

\begin{theorem}
    Consider a system $S$. Let $G_1, G_2$, and $G_3$ be $pe$-graphs representing subsystems of $S$ over the same time interval. Let $\phi \colon G_1 \circlearrow G_2$ and $\psi \colon G_2 \circlearrow G_3$ be homorheisms. Then $\psi \circ \phi \colon G_1 \circlearrow G_3$ is also a homorheism.
    \label{thm: homorheismcomp}
\end{theorem}
\begin{proof}
    Follows from Theorems \ref{thm: PEcomp} and \ref{thm: reflcomp}.
\end{proof}

\subsubsection{Tests for preservation and reflection of closure}

The next result establishes useful conditions for determining the preservation of closure based on structural properties of the induced images of cycles under the weak graph homomorphism.

\begin{theorem}[Preservation test]
    Let $G$ and $H$ be graphs and let $\phi \colon G \to H$ be a weak graph homomorphism. Then $\phi$ preserves closure if and only if for every cycle $C \subseteq G$ either
    \begin{enumerate}[(i)]
        \item $|V(H[\phi(C)])| > 1$, or
        \item $H[\phi(C)]$ is a graph with a single vertex and a single loop.
    \end{enumerate}
    \label{thm: prestest}
\end{theorem}
\begin{proof}
    We prove the contrapositive in both directions, that is we show that $\phi$ does not preserve closure if and only if there exists a cycle $C \subseteq G$ such that $H[\phi(C)]$ is a graph with a single vertex and no edges. If $\phi$ does not preserve closure then there exists a cycle $C \subseteq G$ such that the graph $H[\phi(C)]$ does not contain a cycle, so $H[\phi(C)]$ must consist of a single vertex and no edges. Conversely, if there exists a cycle $C \subseteq G$ such that the graph $H[\phi(C)]$ has a single vertex and no edges then $H[\phi(C)]$ contains no cycles, hence $\phi$ does not preserve closure.
\end{proof}

Similarly, we can check whether a given weak graph homomorphism reflects closure by examining whether any cycle in the source graph has the \textit{entire} target graph as its induced image.

\begin{theorem}[Reflection test]
    Let $G$ and $H$ be graphs and let $\phi \colon G \to H$ be a weak graph homomorphism. If there exists a cycle $C \subseteq G$ such that $H[\phi(C)] = H$ then $\phi$ reflects closure.
    \label{thm: refltest}
\end{theorem}

\begin{proof}
    If $C\subseteq G$ is a cycle such that $H[\phi(C)] = H$ then for every cycle $D$ in $H$ we have $D \subseteq H = H[\phi(C)]$.
\end{proof}

\subsection{The $pe$-graph $I_\mathcal{P}$}

In this section we discuss in greater detail the $pe$-graph $I_\mathcal{P}$ from Figure~\ref{fig: intermediate} in Section \ref{sect: FA}, and prove that the maps $\phi_5 \colon I_\mathcal{P} \circlearrow F_A$ and $\phi_6 \colon I_\mathcal{P} \circlearrow A_3$ are homorheisms. We abbreviate the processes described in $I_\mathcal{P}$, $F_A$, and $A_3$ through the correspondence in Table \ref{tab: abbreviations}.

\begin{table}[h]
    \small
    \centering
    \setlength{\tabcolsep}{0.5cm}
    \caption{Abbreviations for the processes in the $pe$-graphs $I_\mathcal{P}$, $F_A$, and $A_3$ from Figure~\ref{fig: intermediate} in Section~\ref{sect: FA}.}
    \begin{tabular}{c l}
        \toprule
        Abbreviation & Process \\
        \midrule
        CC & Covalent Chemistry \\[5pt]
        PF & Protein Folding \\[5pt]
        IT & Ion Transport \\[5pt]
        MM & Membrane Maintenance \\[5pt]
        NT & Nutrient Transport \\[5pt]
        F & Fabrication \\[5pt]
        A & Assembly \\[5pt]
        I & Maintenance of the intracellular milieu \\[5pt]
        M & Metabolic Reaction Network \\[5pt]
        B & Boundary Maintenance \\[5pt]
        T & Transport \\
        \bottomrule
    \end{tabular}
    \label{tab: abbreviations}
\end{table}

We first confirm that each of the edges in $I_\mathcal{P}$ is a direct enablement:
\begin{itemize}
    \item NT $\to$ CC. \textbf{(Enablement):} Once nutrients enter the cell via NT, they are metabolised through the process of CC. The latter would not occur without the former as there would be no reactants for all of the metabolic reactions in CC. \textbf{(Direct):} As nutrients enter the cell they immediately begin interacting with enzymes and other molecules.
    \item CC $\to$ MM. \textbf{(Enablement):} CC produces the lipids that are necessary to form the cell membrane. \textbf{(Direct):} Membrane maintenance includes the trafficking of lipids to the membrane. As lipids are synthesised within the cytoplasm they immediately begin to self-assemble and commence their journey to the cell membrane.
    \item CC $\to$ PF. \textbf{(Enablement):} CC produces the polypeptides that are needed for protein folding to occur. \textbf{(Direct):} As polypeptides are produced by the ribosome they interact with the local intracellular environment and begin to take on secondary structural features through the process of cotranslational folding.
    \item PF $\to$ CC. \textbf{(Enablement):} Folded enzymes are needed to catalyze the reactions in CC. \textbf{(Direct):} Protein folding is a continuous process, since proteins shift between conformations as they diffuse around the cell and participate in metabolic reactions. 
    \item PF $\to$ IT. \textbf{(Enablement):} Folded ion transporters are needed to transport ions across the membrane. \textbf{(Direct):} Membrane transporters constantly have to maintain their conformation as they interact with the intracellular milieu, the extracellular environment, and the hydrophobic interior of the cell membrane.
    \item PF $\to$ NT. \textbf{(Enablement):} Folded membrane transporters (like GLUTs) are needed to transport nutrients across the membrane.
    \textbf{(Direct):} As for PF $\to$ IT.
    \item PF $\to$ MM. \textbf{(Enablement):} Folded enzymes (like flippase and scramblase) are needed to keep the cell membrane regulated and intact. \textbf{(Direct):} As for PF $\to$ IT. 
    \item MM $\to$ PF. \textbf{(Enablement):} The maintenance of the membrane ensures that ions remain within the cell, thus controlling the composition of the intracellular milieu which mediates PF. Without MM, a protein's folding environment would be wildly unregulated, likely causing it to fold incorrectly. \textbf{(Direct):} Proteins that fold near the boundary of the cell will interact with ions that have just bounced off the cell membrane. 
    \item MM $\to$ NT. \textbf{(Enablement):} Without a membrane there is no distinction between the interior and exterior of the cell and therefore transport cannot occur.
    \textbf{(Direct):} Membrane transport proteins need to exist within the lipid bilayer in order to function properly. 
    \item MM $\to$ IT. \textbf{(Enablement):} As for MM $\to$ NT. \textbf{(Direct):} As for MM $\to$ NT.
    \item MM $\to$ CC. \textbf{(Enablement):} The membrane ensures that all reactants and catalysts involved in CC remain close enough together to form a functional metabolic reaction network.
    \textbf{(Direct):} The bouncing of particular molecules against the boundary of the cell membrane causes reactants and catalysts to interact.
    \item IT $\to$ PF. \textbf{(Enablement):} Ions are needed to ensure a proper environment for protein folding. \textbf{(Direct):} Some protein folding occurs near ion channels, allowing for some interaction between these two processes there.
\end{itemize}

We will now show that $\phi_5$ and $\phi_6$ are homorheisms. Recall that a homorheism is a weak graph homomorphism that preserves closure and reflects closure. Note that $\phi_5 \colon V(I_\mathcal{P}) \to V(F_A)$ is the vertex map given by $CC \mapsto F$, $NT \mapsto F$, $PF \mapsto A$, $MM \mapsto A$, and $IT \mapsto M$; and $\phi_6 \colon V(I_\mathcal{P}) \to V(A_3)$ is the vertex map given by $CC \mapsto M$, $NT \mapsto T$, $PF \mapsto M$, $MM \mapsto B$, and $IT \mapsto T$. Both $\phi_5$ and $\phi_6$ are weak graph homomorphisms, so it suffices to show that they preserve closure and reflect closure. For reference, all of the cycles in $I_\mathcal{P}$ are shown in Figure \ref{fig: Ipcycles}.

\begin{figure}[ht]
    \centering
    \begin{tabular}{c c c}
    \setlength{\tabcolsep}{0.3cm}
        \\
        % -------ROW 1 START -----
        
        % ROW 1 COL 1
        \begin{tikzpicture}[node distance={0.5cm}]
        % Nodes
        \node[half fill={blue!20}{green!20}] (CC) at (0, 1)     {\footnotesize CC};
        \node[half fill={orange!30}{violet!30}] (MM) at  (1.6,1)     {\footnotesize MM};
        \node[] (label) at  (0.8,-0.2)     {\footnotesize $C_1$};
        % Arrows
        \draw[->, shorten <= 0.2cm, shorten >= 0.2cm, transform canvas={yshift=-0.6ex}] (MM) -- (CC);
        \draw[->, shorten <= 0.2cm, shorten >= 0.2cm, transform canvas={yshift=0.6ex}] (CC) --  (MM);
        % Border
        \node [draw, fit=(CC)(MM), inner sep = 0.4cm] {};
        \end{tikzpicture} &
        
        % ROW 1 COL 2
        \begin{tikzpicture}[node distance={0.5cm}]
        % Nodes
        \node[half fill={blue!20}{green!20}] (CC) at (0, 1)     {\footnotesize CC};
        \node[half fill={orange!30}{green!20}] (PF) at  (1.5,1)     {\footnotesize PF};
        \node[] (label) at  (0.8,-0.2)     {\footnotesize $C_2$};
        % Arrows
        \draw[->, shorten <= 0.2cm, shorten >= 0.2cm, transform canvas={yshift=-0.6ex}] (PF) -- (CC);
        \draw[->, shorten <= 0.2cm, shorten >= 0.2cm, transform canvas={yshift=0.6ex}] (CC) --  (PF);
        % Border
        \node [draw, fit=(CC)(PF), inner sep = 0.4cm] {};
        \end{tikzpicture} &

        % ROW 1 COL 3
        \begin{tikzpicture}[node distance={0.5cm}]
        % Nodes
        \node[half fill={orange!30}{green!20}] (PF) at (0, 1)     {\footnotesize PF};
        \node[half fill={gray!25}{red!25}] (IT) at  (1.5,1)     {\footnotesize IT};
        \node[] (label) at  (0.8,-0.2)     {\footnotesize $C_3$};
        % Arrows
        \draw[->, shorten <= 0.2cm, shorten >= 0.2cm, transform canvas={yshift=-0.6ex}] (PF) -- (IT);
        \draw[->, shorten <= 0.2cm, shorten >= 0.2cm, transform canvas={yshift=0.6ex}] (IT) --  (PF);
        % Border
        \node [draw, fit=(PF)(IT), inner sep = 0.4cm] {};
        \end{tikzpicture} \\[10pt]
        
        % ROW 2 COL 1
        \begin{tikzpicture}[node distance={0.5cm}]
        % Nodes
        \node[half fill={orange!30}{green!20}] (PF) at (0, 1)     {\footnotesize PF};
        \node[half fill={orange!30}{violet!30}] (MM) at  (1.6,1)     {\footnotesize MM};;
        \node[] (label) at  (0.8,-0.2)     {\footnotesize $C_4$};
        \node[] (phantom) at (0.8, -0.95) {};
        % Arrows
        \draw[->, shorten <= 0.2cm, shorten >= 0.2cm, transform canvas={yshift=-0.6ex}] (PF) -- (MM);
        \draw[->, shorten <= 0.2cm, shorten >= 0.2cm, transform canvas={yshift=0.6ex}] (MM) --  (PF);
        % Border
        \node [draw, fit=(PF)(MM), inner sep = 0.4cm] {};
        \end{tikzpicture} &
        
        % ROW 2 COL 2
        \begin{tikzpicture}[node distance={0.5cm}]
        % Nodes
        \node[half fill ={orange!30}{green!20}] (PF) at  (0,0)     {\footnotesize PF};
        \node[half fill={blue!20}{green!20}] (CC) at (1,1.3)     {\footnotesize CC};
        \node[half fill={orange!30}{violet!30}] (MM) at  (2,0)     {\footnotesize MM};
        \node[] (label) at  (1,-1.2)     {\footnotesize $C_5$};
        % Arrows
        \draw[<-, shorten <= 0.2cm, shorten >= 0.2cm] (PF) -- (CC);
        \draw[<-, shorten <= 0.2cm, shorten >= 0.2cm] (CC) --  (MM);
        \draw[<-, shorten <= 0.2cm, shorten >= 0.2cm] (MM) --  (PF);
        % Border
        \node [draw, fit=(MM)(CC)(PF), inner sep = 0.4cm] {};
        \end{tikzpicture} &
        
        % ROW 2 COL 3
        \begin{tikzpicture}[node distance={0.5cm}]
        % Nodes
        \node[half fill={orange!30}{green!20}] (PF) at  (0,0)     {\footnotesize PF};
        \node[half fill={blue!20}{green!20}] (CC) at (1,1.3)     {\footnotesize CC};
        \node[half fill ={orange!30}{violet!30}] (MM) at  (2,0)     {\footnotesize MM};
        \node[] (label) at  (1,-1.2)     {\footnotesize $C_6$};
        % Arrows
        \draw[->, shorten <= 0.2cm, shorten >= 0.2cm] (PF) -- (CC);
        \draw[->, shorten <= 0.2cm, shorten >= 0.2cm] (CC) --  (MM);
        \draw[->, shorten <= 0.2cm, shorten >= 0.2cm] (MM) --  (PF);
        % Border
        \node [draw, fit=(MM)(CC)(PF), inner sep = 0.4cm] {};
        \end{tikzpicture} \\[10pt]

        % ROW 3 COL 1
        \begin{tikzpicture}[node distance={0.5cm}]
        % Nodes
        \node[half fill={blue!20}{red!25}] (NT) at  (0,0)     {\footnotesize NT};
        \node[half fill={blue!20}{green!20}] (CC) at (1,1.3)     {\footnotesize CC};
        \node[half fill ={orange!30}{violet!30}] (MM) at  (2,0)     {\footnotesize MM};
        \node[] (label) at  (1,-1.2)     {\footnotesize $C_7$};
        % Arrows
        \draw[->, shorten <= 0.2cm, shorten >= 0.2cm] (NT) -- (CC);
        \draw[->, shorten <= 0.2cm, shorten >= 0.2cm] (CC) --  (MM);
        \draw[->, shorten <= 0.2cm, shorten >= 0.2cm] (MM) --  (NT);
        % Border
        \node [draw, fit=(NT)(CC)(MM), inner sep = 0.4cm] {};
        \end{tikzpicture} &
        
        % ROW 3 COL 2
        \begin{tikzpicture}[node distance={0.5cm}]
        % Nodes
        \node[half fill={blue!20}{red!25}] (NT) at  (0,0)     {\footnotesize NT};
        \node[half fill={blue!20}{green!20}] (CC) at (1,1.3)     {\footnotesize CC};
        \node[half fill ={orange!30}{green!20}] (PF) at  (2,0)     {\footnotesize PF};
        \node[] (label) at  (1,-1.2)     {\footnotesize $C_8$};
        % Arrows
        \draw[->, shorten <= 0.2cm, shorten >= 0.2cm] (NT) -- (CC);
        \draw[->, shorten <= 0.2cm, shorten >= 0.2cm] (CC) --  (PF);
        \draw[->, shorten <= 0.2cm, shorten >= 0.2cm] (PF) --  (NT);
        % Border
        \node [draw, fit=(NT)(CC)(PF), inner sep = 0.4cm] {};
        \end{tikzpicture} &
        
        % ROW 3 COL 3
        \begin{tikzpicture}[node distance={0.5cm}]
        % Nodes
        \node[half fill={orange!30}{green!20}] (PF) at  (0,0)     {\footnotesize PF};
        \node[half fill={gray!25}{red!25}] (IT) at (1,1.3)     {\footnotesize IT};
        \node[half fill ={orange!30}{violet!30}] (MM) at  (2,0)     {\footnotesize MM};
        \node[] (label) at  (1,-1.2)     {\footnotesize $C_9$};
        % Arrows
        \draw[->, shorten <= 0.2cm, shorten >= 0.2cm] (PF) -- (MM);
        \draw[->, shorten <= 0.2cm, shorten >= 0.2cm] (MM) --  (IT);
        \draw[->, shorten <= 0.2cm, shorten >= 0.2cm] (IT) --  (PF);
        % Border
        \node [draw, fit=(MM)(CC)(PF), inner sep = 0.4cm] {};
        \end{tikzpicture} \\[10pt]

        % ROW 4 COL 1
        \begin{tikzpicture}[node distance={0.5cm}]
        % Nodes
        \node[half fill={orange!30}{green!20}] (PF) at  (0,0)     {\footnotesize PF};
        \node[half fill={blue!20}{green!20}] (CC) at (2,1.8)     {\footnotesize CC};
        \node[half fill ={orange!30}{violet!30}] (MM) at  (2,0)     {\footnotesize MM};
        \node[half fill={blue!20}{red!25}] (NT) at  (0,1.8)     {\footnotesize NT};
        \node[] (label) at  (1,-1.2)     {\footnotesize $C_{10}$};
        % Arrows
        \draw[->, shorten <= 0.2cm, shorten >= 0.2cm] (PF) -- (NT);
        \draw[->, shorten <= 0.2cm, shorten >= 0.2cm] (NT) --  (CC);
        \draw[->, shorten <= 0.2cm, shorten >= 0.2cm] (CC) --  (MM);
        \draw[->, shorten <= 0.2cm, shorten >= 0.2cm] (MM) --  (PF);
        % Border
        \node [draw, fit=(MM)(CC)(PF)(NT), inner sep = 0.4cm] {};
        \end{tikzpicture} &
        
        % ROW 4 COL 2
        \begin{tikzpicture}[node distance={0.5cm}]
        % Nodes
        \node[half fill={gray!25}{red!25}] (IT) at  (0,1.8)     {\footnotesize IT};
        \node[half fill={orange!30}{green!20}] (PF) at  (2,1.8)     {\footnotesize PF};
        \node[half fill={blue!20}{green!20}] (CC) at (2,0)     {\footnotesize CC};
        \node[half fill ={orange!30}{violet!30}] (MM) at  (0,0)     {\footnotesize MM};
        \node[] (label) at  (1,-1.2)     {\footnotesize $C_{11}$};
        % Arrows
        \draw[->, shorten <= 0.2cm, shorten >= 0.2cm] (IT) -- (PF);
        \draw[->, shorten <= 0.2cm, shorten >= 0.2cm] (PF) --  (CC);
        \draw[->, shorten <= 0.2cm, shorten >= 0.2cm] (CC) --  (MM);
        \draw[->, shorten <= 0.2cm, shorten >= 0.2cm] (MM) --  (IT);
        % Border
        \node [draw, fit=(MM)(CC)(PF)(IT), inner sep = 0.4cm] {};
        \end{tikzpicture} & 

        % ROW 4 COL 3
        \begin{tikzpicture}[node distance={0.5cm}]
        % Nodes
        \node[half fill={blue!20}{red!25}] (NT) at  (0,1.8)     {\footnotesize NT};
        \node[half fill={orange!30}{green!20}] (PF) at  (2,0)     {\footnotesize PF};
        \node[half fill={blue!20}{green!20}] (CC) at (2,1.8)     {\footnotesize CC};
        \node[half fill ={orange!30}{violet!30}] (MM) at  (0,0)     {\footnotesize MM};
         \node[] (label) at  (1,-1.2)     {\footnotesize $C_{12}$};
        % Arrows
        \draw[->, shorten <= 0.2cm, shorten >= 0.2cm] (NT) -- (CC);
        \draw[->, shorten <= 0.2cm, shorten >= 0.2cm] (CC) --  (PF);
        \draw[->, shorten <= 0.2cm, shorten >= 0.2cm] (PF) --  (MM);
        \draw[->, shorten <= 0.2cm, shorten >= 0.2cm] (MM) --  (NT);
        % Border
        \node [draw, fit=(MM)(CC)(PF)(NT), inner sep = 0.4cm] {};
        \end{tikzpicture} \\[20pt]

        &
        
        % ROW 5 COL 2
        \begin{tikzpicture}[node distance={0.5cm}]
        % Nodes
        \node[half fill={blue!20}{red!25}] (NT) at  (1.5,2.8)     {\footnotesize NT};
        \node[half fill={blue!20}{green!20}] (CC) at (3,1.6)     {\footnotesize CC};
        \node[half fill ={orange!30}{violet!30}] (MM) at  (2.4,0)     {\footnotesize MM};
        \node[half fill={gray!25}{red!25}] (IT) at  (0.6,0)     {\footnotesize IT};
        \node[half fill={orange!30}{green!20}] (PF) at  (0,1.6)     {\footnotesize PF};
        \node[] (label) at  (1.5,-1.2)     {\footnotesize $C_{13}$};
        % Arrows
        \draw[->, shorten <= 0.2cm, shorten >= 0.2cm] (NT) -- (CC);
        \draw[->, shorten <= 0.2cm, shorten >= 0.2cm] (CC) --  (MM);
        \draw[->, shorten <= 0.2cm, shorten >= 0.2cm] (MM) --  (IT);
        \draw[->, shorten <= 0.2cm, shorten >= 0.2cm] (IT) --  (PF);
        \draw[->, shorten <= 0.2cm, shorten >= 0.2cm] (PF) --  (NT);
        \node [draw, fit=(IT)(MM)(CC)(PF)(NT), inner sep = 0.4cm] {};
        \end{tikzpicture} \\[18pt]
    \end{tabular}
    \caption{The cycles of the $pe$-graph $I_\mathcal{P}$. The colour coding corresponds to the maps $\phi_5$ and $\phi_6$ described in Figure~\ref{fig: intermediate} in Section \ref{sect: FA}.}
    \label{fig: Ipcycles}
\end{figure}
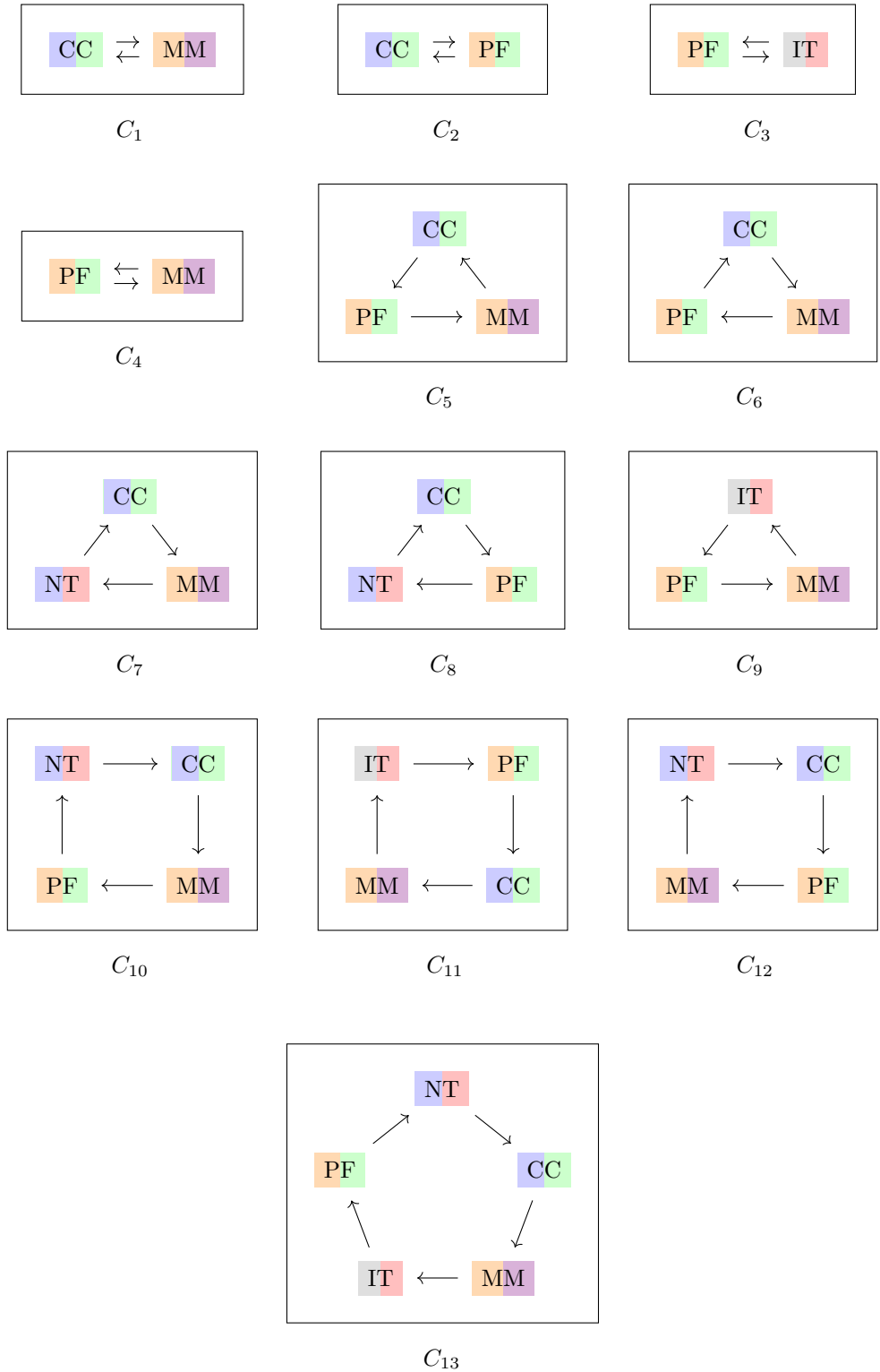

\begin{theorem}
     The weak graph homomorphisms $\phi_5$ and $\phi_6$ preserve closure.
     \label{thm: intermedpres}
\end{theorem}
\begin{proof}
    By inspecting the colourings of each cycle in $I_\mathcal{P}$ from Figure~\ref{fig: Ipcycles}, note that the induced image of every cycle under $\phi_5$ and $\phi_6$ contains more than one vertex, except for the graph $A_3[\phi_6(C_2)]$ which has the single vertex $M$ and a single loop, and except for the graph $F_A[\phi_5(C_4)]$ which has the single vertex $A$ and a single loop. It therefore follows from Theorem \ref{thm: prestest} that $\phi_5$ and $\phi_6$ preserve closure.
\end{proof}

\begin{theorem}
     The weak graph homomorphisms $\phi_5$ and $\phi_6$ reflect closure.
     \label{thm: intermedrefl}
\end{theorem}
\begin{proof}
    Since $F_A[\phi_5(C_{13})] = F_A$ and $A_3[\phi_6(C_{13})] = A_3$, it follows from Theorem \ref{thm: refltest} that $\phi_5$ and $\phi_6$ reflect closure.
\end{proof}

\end{document}